\documentclass[a4paper,11pt]{article}
\usepackage[T1]{fontenc}
\usepackage[latin9]{inputenc}
\usepackage{amsmath}
\usepackage{amssymb}
\usepackage{mathtools}
\usepackage{mathtools}
\usepackage{amsthm}

\usepackage{appendix}
\usepackage{hyperref}

\newtheorem{proposition}{Proposition}[section]
\newtheorem{lemma}{Lemma}[section]

\newtheorem{theorem}{Theorem}[section]
\newtheorem{corollary}{Corollary}[section]
\newtheorem{definition}{Definition}[section]
\newtheorem{remark}{Remark}[section]
\newtheorem{observation}{Observation}[section]

\makeatletter
\@ifundefined{date}{}{\date{}}
\makeatother
\usepackage{fullpage}
\usepackage{graphicx}
\graphicspath{{figures/}}

\usepackage[]{color-edits}

\addauthor{MS}{red}

\addauthor{MF}{blue}

\addauthor{TE}{magenta}

\usepackage{babel}
\usepackage{url}
\usepackage{algorithm} 
\usepackage{algpseudocode}
\usepackage{bbold}
\usepackage{natbib}
\usepackage{thm-restate}

\newcommand{\cclass}{\mathcal{T}}
\newcommand{\argmax}{\text{argmax}}

\newcommand{\reals}{\mathbb{R}}
\newcommand{\E}{\mathbb{E}}

\newcommand{\indicator}{\mathbb{1}}
\newcommand{\relint}{\text{relint}}
\newcommand{\actions}{A}
\newcommand{\outcomes}{[m]}
\newcommand{\prob}{p}
\newcommand{\strat}{\pi}
\newcommand{\finaloutcome}{\omega}
\newcommand{\bestresponse}{\pi^\star}

\newcommand{\boxes}{B}
\newcommand{\prev}{\text{prev}}
\newcommand{\linz}[1]{z^{#1}}

\newcommand{\haltgreater}{\tau^{>}}
\newcommand{\haltsmaller}{\tau^{<}}
\newcommand{\actionspreorder}{P_{\actions}}
\newcommand{\outcomespreorder}{P_{\outcomes}}
\newcommand{\epsiloncontract}{t^{\varepsilon}}
\newcommand{\contractregion}{[0, L]^m}

\begin{document}
\title{Contract Design for Sequential Actions
\thanks{This project has been partially funded by the European Research Council (ERC) under the European Union's Horizon 2020 research and innovation program (grant agreement No. 866132), by an Amazon Research Award, by the NSF-BSF (grant number 2020788), and by a grant from TAU Center for AI and Data Science (TAD). Tomer Ezra is supported by the Harvard University Center of Mathematical Sciences and Applications.}
} 

\author{
Tomer Ezra\thanks{ Harvard University, USA. Email: tomer@cmsa.fas.harvard.edu}\and Michal Feldman\thanks{Tel Aviv University, Israel. Email: mfeldman@tauex.tau.ac.il}\and Maya Schlesinger\thanks{{Tel Aviv University, Israel. Email: mayas1@mail.tau.ac.il}}
}
\maketitle

\begin{abstract}

We introduce a novel model of contracts with combinatorial actions that accounts for sequential and adaptive agent behavior. 
As in the standard model, a principal delegates the execution of a costly project to an agent. There are $n$ actions, each one incurring a cost to the agent and inducing a probability distribution over $m$ outcomes; each outcome generates some reward for the principal. 
The principal incentivizes the agent through a {\em contract} that specifies a payment for each potential outcome. 
Unlike the standard model, the agent chooses actions {\em sequentially}. Following each action, the agent observes the realized outcome, and decides whether to stop or continue with another action. Upon halting, the agent chooses one of the realized outcomes, which determines both his payment and the principal's reward.
This model captures common scenarios where the agent can make multiple attempts in the course of executing a project.

We study the optimal contract problem in this new setting, namely the contract that maximizes the principal's utility. We first observe that the agent's problem --- (adaptively) finding the sequence of actions that maximizes his utility for a given contract --- is equivalent to the well-known Pandora's Box problem. Using this insight, we provide algorithms and hardness results for the optimal contract problem, under both independent and correlated actions, and for both linear and general contracts. 
For independent actions, we provide a poly-time algorithm for the optimal linear contract, and establish that finding the optimal general contract is NP-hard. In cases where the number of outcomes is constant, we devise a poly-time algorithm even for the optimal general contract. 
For correlated actions, we find that, for both linear and general contracts, approximating the optimal contract within any constant ratio is NP-hard.
\end{abstract}
\newpage

\section{Introduction}

Contract theory is a pillar in microeconomics, addressing the problem of motivating agents to exert effort when their actions are not directly observable. This challenge is explored through the principal-agent model of \cite{holmstrom1979moral} and \cite{grossman-hart-1983}.

In this model, a principal delegates the execution of a project to an agent, who chooses an action from a set of $n$ actions, each associated with a given cost and a probability distribution over $m$ possible outcomes. Each one of the outcomes generates some reward to the principal.
The principal cannot observe the action chosen by the agent, but only the obtained outcome.
To motivate the agent to take the desired action, the principal designs a contract, namely a payment scheme $t:\outcomes \rightarrow \reals_{\geq 0}$ that specifies some payment to the agent for every observed outcome.
Given a contract, the agent chooses an action that maximizes his expected utility, calculated as the expected payment minus the cost.

The optimal contract problem is finding the contract that maximizes the principal's expected utility, calculated as the expected reward from the incentivized action minus the expected payment.
This problem is poly-time solvable using linear programming \citep{grossman-hart-1983}.

Contract design has traditionally been studied within the field of economics. However, recent years have seen a growing interest in exploring contract design from an algorithmic perspective, for a recent survey, see \cite{survey}.
This shift has broadened the study of contract design to include various algorithmic and computational perspectives, addressing the complexities of real-world scenarios.
A key model within this realm is the {\em combinatorial contract} model, where the agent can choose combinations of actions, rather than a single action \citep{combinatorial-contracts, duetting2023combinatorial, ezra2023approximability,DuttingEFK25}.

A common assumption in all previous works on combinatorial contracts is that the agent selects all actions simultaneously. However, in many real-world scenarios, the agent chooses actions sequentially, observing intermediate outcomes before deciding on the next action.
In this work, we propose a framework capturing such scenarios and study the optimal contract problem within this framework.

\paragraph{The sequential multi-action model.}
We introduce a contract design model, where an agent performs multiple actions sequentially, potentially adjusting their actions based on the outcomes of prior efforts. Our model captures a wide array of real-world scenarios. 
For instance, in a typical recruitment process, a company (the principal) engages a recruiter (the agent) to fill a vacancy. The recruiter may interview multiple candidates sequentially, observing their competence before deciding to proceed with others.
Similarly, a freelance data scientist (the agent) might test different methodologies, presenting only the most effective solution to the client (the principal). Other examples include a teacher employing diverse techniques to prepare students for an exam, a doctor exploring treatments for a patient, or a researcher evaluating multiple questions before selecting one for publication.

In our model, there are $n$ actions, each associated with a cost to the agent, and $m$ possible outcomes, each associated with a reward to the principal.
Every action induces a probability distribution over outcomes, which may be independent or correlated across different actions (see details below).

The principal begins by offering a contract that specifies a payment for each potential outcome. Given the contract, the agent's {\em strategy} proceeds iteratively: in each iteration, the agent selects an action, incurs its cost, and observes its realized outcome. The agent then decides whether to continue with another action or to halt.
Upon halting, the agent selects one of the outcomes obtained during the process as the {\em final outcome}, which determines the principal's reward. The principal observes only the final outcome, and pays the agent according to the contract's payment for that outcome.\footnote{In a concurrent and independent work, \cite{contract_pandora} introduces a related model with sequential actions, which deviates from the common hidden-action assumption, namely, the principal observes both the final outcome and the action that led to that outcome. Naturally, this gives the principal more power, as reflected in the results, see more details in Section~\ref{sec:related}.}
The agent's utility is then calculated as the payment received from the principal minus the total costs of the actions undertaken.

The principal's utility is the difference between the expected reward of the final outcome and the expected payment to the agent.
The {\em optimal contract} is the contract that maximizes the principal's utility, given that the agent best responds to the contract. 
In this paper, we study the computational complexity of the optimal (or approximately-optimal) contract problem in this setting.

Of special interest is the class of {\em linear contracts}, where the principal pays the agent a fixed fraction $\alpha$ of the reward. Linear contracts are common in many real-life settings, partly since they are easy to describe and apply. Recent work has discovered additional desired properties of linear contracts, such as max-min optimality \cite{caroll, dutting2019simple}, ambiguity-proofness \cite{ambiguous}, and learnability \cite{sample_complexity}.

\subsection{Our Contributions}
\label{sec:our-results}

Our contribution encompasses both modeling novelty and new results for the optimal contract problem in our setting.

\subsubsection{Modeling Novelty}
We introduce a new framework to capture scenarios where an agent takes actions sequentially, with its distinctiveness highlighted in the following aspects\footnote{Some of these aspects are also explored in \cite{contract_pandora}, introduced concurrently and independently of our work.}:
\begin{itemize}
    \item New combinatorial interaction paradigm:
All prior research on combinatorial actions 
relies on a set function that maps each action set to a success probability. In contrast, our work introduces a fundamentally new framework for how actions interact combinatorially. 
We view this as an exciting first step toward exploring contracts in other combinatorial structures, moving beyond the established paradigm.
    \item Extension to multiple outcomes:
While earlier studies on combinatorial actions 
focused on binary-outcome models (or, equivalently, linear contracts beyond binary-outcome), we extend this framework to incorporate multiple outcomes (even beyond linear contracts), providing a richer and more realistic model.
    \item Sequential and adaptive decision-making:
Our model allows the agent to take actions sequentially, adapting his subsequent decisions based on previously realized outcomes. This is a significant departure from previous work, where agents executed all selected actions simultaneously without such adaptive flexibility.
\end{itemize}

\subsubsection{Our Results}
We provide structural and complexity results for the optimal contract under independent and correlated actions, with respect to both linear and general contracts.

\paragraph{The agent's best-response problem.}
Before addressing the principal's optimal contract problem, we first explore the agent's best-response problem: given a contract, determining the agent's optimal (sequential) strategy. A key observation here is that the agent's best-response problem reduces to solving a Pandora's Box instance --- an online search problem with exploration costs \citep{weitzman1978optimal}. For further details, see Section~\ref{sec:pandora}.

\begin{observation}
    The agent's problem --- finding a strategy that maximizes his utility given a contract designed by the principal --- reduces to the Pandora's Box problem. 
\end{observation}

Notably, while this is a fundamental observation, the results for the agent's best-response problem do not directly imply either positive or negative results for the principal's problem.

\citet{weitzman1978optimal} provides a concise description of the optimal strategy for the Pandora's Box problem in cases where the values are independently distributed. This solution directly applies to the agent's strategy in the independent-action setting (but not in the correlated-action setting). We now describe the implications of Weitzman's algorithm for the agent's problem using our terminology.
Given a contract (specifying a payment for every outcome), the agent's best response involves calculating a {\em reservation value} for every action $i$, based on its cost and the induced payment distribution. The agent then executes actions in non-decreasing order of these reservation values, halting when the highest payment across all revealed outcomes exceeds (or meets) the subsequent reservation value.
Our results for the independent-action model rely on analyzing how the reservation values depend on the contract proposed by the principal.

We are now ready to introduce our results for the (principal's) optimal contract problem.

\paragraph{Independent-action, linear contracts.} Our first result for the independent-action model is a poly-time algorithm that finds the optimal linear contract.

\vspace{0.1in}
\noindent {\bf Theorem} {(Linear contracts, Theorem~\ref{thm:polytime_linear_contract})}
    In the independent action model, the optimal linear contract can be computed in polynomial (in $n$ and $m$) time.
\vspace{0.1in}

{\bf Proof sketch.} 
Recall that a linear contract is defined by a parameter $\alpha\in [0,1]$. Consequently, each action's reservation value is expressed as a function of $\alpha$. To prove this theorem, we show that this function is a convex, piecewise-linear function with at most $m$ segments. From this, we deduce that the agent's best response may change only a polynomial number of times. By identifying the {\em transition values} of $\alpha$ --- the values at which the agent's best response shifts --- we obtain a polynomially bounded set of candidates for the optimal linear contract, from which the best option can be selected.

\paragraph{Independent-action, general contracts.} 
We now turn our attention to general contracts. On the positive side, we devise an algorithm that computes the optimal general contract in polynomial time, assuming $m$ (the number of outcomes) is constant.

\vspace{0.1in}
\noindent {\bf Theorem} {(General contracts with constant $m$, Theorem~\ref{thm:opt_general_contract})}
    In the independent action model for a constant number of outcomes $m$, the optimal general contract can be computed in polynomial time.
\vspace{0.1in}

{\bf Proof sketch.}
Computing an optimal general contract in a combinatorial setting is an entirely unexplored problem.
Previous approaches to computing optimal (general or linear) contracts can be categorized into two methods: (1) LP-based: iterating over all possible actions and using linear programming to compute the best contract that incentivizes each action, or (2) Iterating over critical points: for linear contracts, identifying critical points in the single-dimensional contract space where the agent's best response changes and selecting the optimal point. 
However, neither method extends to general contracts in combinatorial settings. The LP-based approach fails due to the exponential number of agent strategies to consider.
Iterating over critical points is inapplicable because general contracts involve optimizing over a multi-dimensional space, lacking monotonicity properties that have been essential in previous work on linear contracts.

To address these challenges, we propose a novel approach that combines key elements from both existing methods. Specifically, we first characterize the agent strategies that can be incentivized by a contract (drawing inspiration from the critical-point approach) and then identify the optimal contract to incentivize them (building on ideas from the LP-based method).
In our approach, we think of a contract as a vector in $\reals_{\ge 0}^m$ and construct a hyperplane arrangement of polynomial size, ensuring that within any face of the arrangement, the set of all strategies maximizing the agent's utility remains unchanged.
Analogous to the LP-based method, we show that the optimal contract is located at a vertex of the hyperplane arrangement. Additionally, inspired by the critical-point method and utilizing the fact that the hyperplane arrangement is of polynomial size in a vector space of a constant dimension, we deduce that it has a polynomial number of vertices. By iterating over these vertices, we can efficiently identify the optimal contract.

On the negative side, we show that computing the optimal general contract for an arbitrary number of outcomes is NP-hard. This is done using a reduction from 
\textsc{Partition}.

\vspace{0.1in}
\noindent {\bf Theorem} {(General contracts with arbitrary $m$, Theorem~\ref{thm:independent_np_hard})}
    In the independent action model for  a general number of outcomes $m$, the optimal general contract problem is NP-hard.
\vspace{0.1in}

\paragraph{Independent-action, linear vs. general contracts.} 

In Appendix~\ref{app:independent_lin_opt_ratio}, we show that, as in the classical setting studied by \citet{dutting2019simple} (where the agent is restricted to a single action), the gap between the principal's utility under the best linear contract and the best general contract can be as large as 
$\Omega(n)$.
Additionally, we establish an upper bound of $O(n^2m)$ on this gap by leveraging our analysis of linear contracts. Specifically, we show that linear contracts admit at most $O(n^2m)$
agent best responses, which, as shown by \citet{dutting2019simple}, provides an upper bound on the gap between the best linear contract and the best general contract (or even the social welfare).

\paragraph{The correlated-action model.}
In the correlated-action model, we establish a hardness of approximation result, showing that even in the binary-outcome case (where a project can either succeed or fail) it is NP-hard to approximate the optimal contract within any constant. This hardness also rules out a poly-time approximation of the optimal linear contract, since in the binary-outcome case the optimal contract is linear.

Notably, the NP-hardness established by \citet{chawla2020pandora} for approximating the Pandora's Box problem with correlated boxes applies to the agent's problem, but does not in itself imply NP-hardness for the (principal's) optimal contract problem, which is what the following theorem shows.

\vspace{0.1in}
\noindent {\bf Theorem} {(Theorem~\ref{thm:corr_contract_hardness})}
    In the correlated action model, it is NP-hard to approximate the optimal (linear or general) contract within any constant, even in the binary-outcome case.
\vspace{0.1in}

{\bf Proof idea.}
We establish the hardness of the 
optimal contract problem by observing that, in the binary-outcome setting, 
the complex multi-dimensional distribution of the actions' stochastic outcomes can be effectively represented by a set function $f$,
which maps each set of actions to the probability that at least one leads to a project success. We refer to this class of functions ``correlated OR'' and show its equivalence to the class of coverage functions (an equivalence that may be of independent interest).
Our hardness result builds on an NP-hard promise problem with respect to coverage functions, introduced in \citep{ezra2023approximability}. By leveraging the equivalence between coverage functions and ``correlated OR'' functions, we construct a reduction from this problem to ours.

\subsection{Related Work}
\label{sec:related}
\paragraph{Comparison to recent related work on multi-action contract design.}

Our model is closely related to two recent studies in the algorithmic contract design literature. 

The first is \cite{combinatorial-contracts}, where the agent can choose to take any subset of his $n$ available actions.
The main difference between our model and this one is that, in that model, the agent makes a static decision regarding the set of actions to take, whereas in our model, the agent takes actions sequentially, and can dynamically adjust his actions based on history. The main result in \cite{combinatorial-contracts} is a poly-time algorithm for the optimal (linear) contract when $f$ is gross-substitutes. 
This model has been further studied by \citet{vuong2023supermodular} and \citet{duetting2023combinatorial}, who generalize the algorithmic result \citet{combinatorial-contracts} from gross-substitutes $f$ to any class of functions that admits an efficient algorithm for the agent's best-response and poly-many ``breakpoints'' in the agent's best-response.
\citet{ezra2023approximability} provide hardness of approximation results for this setting. 

The second is the one proposed by \citet{contract_pandora} who, concurrently and independently to us, study a model where the principal delegates a Pandora's Box problem. 
This model and ours are incomparable. 
On the one hand, in their model, the principal gains more information about the agent's actions compared to our assumptions. 
On the other hand, \citet{contract_pandora} allow the agent to have his own inherent rewards for outcomes.
Indeed, in their paper they show a poly-time algorithm for the optimal contract in the case where all agent's rewards are zero, whereas we show this problem to be NP-hard under our principal's information assumptions.
\citet{contract_pandora} also show a polynomial time algorithm for computing the optimal linear contract in their model, as well as poly-time algorithms for computing the optimal general contract under different assumptions, namely the case of binary-value boxes, and the case of i.i.d. boxes with a single positive prize to the principal. 

\paragraph{Delegated search settings.}

Also relevant to our work are delegated search models considered in \citep{KleinbergK18}. In the most general form, in these models the agent draws samples of outcomes from known distributions and chooses one of the drawn outcomes, which the principal can then accept or reject. These outcomes yield different values to the principal and the agent, misaligning their incentives. To incentivize the agent, rather than designing a contract, \citet{KleinbergK18} assume the principal sets a condition on which outcomes the agent may choose, and show simple mechanisms by the principal that achieve a constant-factor approximation for various models in this setting.

Other studies of delegated search and experimentation include \citet{postl2004delegated}, who considers a two-alternative delegated search setting, \citet{guo2016dynamic}, who considers a delegated experimentation model where the principal offers the agent a menu of strategies from which the agent chooses one, and \citet{mauring2016two}, considering a case where two players play out a search strategy together, one revealing alternatives and one selecting their own preferred revealed outcome.

\paragraph{The Pandora's Box problem. } 
A closely related problem to the agent's problem is the Pandora's Box problem, introduced by \citet{weitzman1978optimal}, describing optimal search with costly inspection. In the Pandora's Box problem there are $n$ boxes each with a stochastic hidden value and an associated cost. The decision maker opens the boxes sequentially. Upon opening a box, the hidden value is revealed and the associated cost is incurred. The goal is to maximize the expected maximum value revealed minus the expected total cost. \citet{weitzman1978optimal} showed a simple optimal strategy for this problem, when the hidden values are distributed independently. 

\citet{chawla2020pandora} showed that in the Pandora's Box problem with correlated values, the optimal non-adaptive strategy is NP-hard to approximate within any constant, even with fully adaptive strategies. For a modified version in which one seeks to minimize the expected value plus total cost they show a poly-time constant-factor approximation algorithm for the optimal partially-adaptive strategy. Further work on the correlated values model focuses on the minimization variant and advances our understanding of approximation algorithms for this setting \citep{chawla2023approximating, gergatsouli2023weitzmans}.

\paragraph{Other combinatorial contract settings. }
Another combinatorial model of contracts is the multi-agent setting, where the principal interacts with multiple agents.
This model was first introduced by \citet{BabaioffFNW12}, who focused on the binary-action, binary-outcome scenario, where agents have their own individual outcomes and some boolean function maps these to the project's outcome. \citet{EmekF12} study this model in the special case where the project succeeds when at least one agent succeeds. 
This model was later extended by \citet{multi-agent-contracts}, who consider a model where the success of the project is determined by a function $f: {\{0,1\}}^{[n]}\rightarrow [0,1]$, mapping the agent's actions to a probability of success. \citet{CastiglioniM023} study a multi-agent setting in which each agent has his own outcome, and the principal's reward depends on all individual outcomes. \citet{dasaratha2023equity} consider a multi-agent setting with graph-based reward functions and continuous effort levels. \citet{cacciamani2024multi} study a multi-agent setting where the entire contract instance (which is exponential in size) is given as input, and study questions such as utility guarantees and randomized contracts.

Additional literature on combinatorial contracts includes settings with exponentially many outcomes \citep{dutting2021complexity} and settings with multiple principals \citep{AlonLST21}.

\section{Model and Preliminaries} \label{sec:model}

We consider a principal-agent setting where the principal delegates the execution of a project to an agent.
In order to execute the project, the agent can take actions in a sequential way from a set $\actions$ of $n$ possible actions, where each action $i \in \actions$ is associated with a cost $c_i \geq 0$.
The execution of the project 
leads to an outcome in a set $\outcomes$ of $m$ possible outcomes, where each outcome $j \in \outcomes$ is associated with a reward for the principal $r(j) \ge 0$.  We assume without loss of generality that $r(1) \le \dots \le r(m)$, and that $r(1) = 0$. We sometimes refer  to $1$ as \emph{the zero outcome}.

Actions lead to outcomes stochastically. Specifically each action $i \in \actions$ is associated with a random variable $X_i$ that attains values within $\outcomes$, denoting that action's outcome.
The variables $\{X_i\}_{i \in \actions}$ are distributed according to some
joint distribution $\mathcal{F}$, and are either independent or correlated. The case where they are independent is referred to as the \emph{independent action model} and the case where they may be correlated is referred to as the \emph{correlated action model}.
If they are independent, then the distribution of each $X_i$ is given by a vector $\{\prob_{ij}\}_{j \in \outcomes}$, where $p_{ij}$ is the probability that $X_i = j$.
If the actions are correlated, then for every realization vector $v \in \outcomes^{n}$, $\prob(v)$ denotes the probability that $v$ is realized.

The agent's choices are {\em hidden} from the principal, who can observe only the project's final outcome, for which she is rewarded.
To incentivize the agent to take actions, the principal designs a contract $t:\outcomes \rightarrow \mathbb{R}_{\ge 0}$ which maps each outcome to a non-negative payment to the agent.

\paragraph{Sequential model.}
In a sequential model, once a contract has been determined by the principal, the agent takes actions sequentially. At each point in time, the agent selects an action $i \in \actions$, incurs the associated cost $c_i$, and observes the outcome that is revealed as a result of this action (i.e., the value of $X_i$). The agent then decides whether to halt or continue with another action (that has not yet been taken).
When the agent halts, he selects one outcome from those revealed, and this selected outcome becomes the final outcome of the project, upon which the agent is paid. 

A strategy of the agent in this scenario, denoted $\strat$, determines which action to take at each point in time, and may choose to halt.
When halting, the agent chooses which of the obtained outcomes is the final project's outcome, for which the principal is rewarded and the agent is paid.
If the agent decides to halt before taking any action, then the zero outcome is the final outcome, which we assume is always revealed.
All decisions may depend on the instance, the contract, the actions taken so far, and their revealed outcomes.

Given an agent's strategy $\strat$, we denote by $S(\strat)$ the random variable depicting the set of actions taken by the agent, and by $\finaloutcome(\strat)$ the final outcome. 

The agent's utility under a contract $t$ and a strategy $\strat$ is the expected payment from the principal minus the expected sum of the costs of the actions taken, 
\[
u_A(t, \strat) = \E_{X_1, \dots, X_n \sim \mathcal{F}}[t(\finaloutcome(\strat)) - \sum_{i\in S(\strat)} c_i].
\]
The principal's utility under a contract $t$ and a strategy $\strat$  is the expected reward minus the expected payment, i.e.,
\[
u_P(t, \strat) = \E_{X_1, \dots, X_n \sim \mathcal{F}}[r(\omega(\strat)) - t(\omega(\strat))].
\]

\paragraph{Best response.} The agent's best response to a contract $t$, denoted $\bestresponse(t)$, is a strategy that maximizes the agent's utility, namely $\bestresponse(t) \in \argmax_{\strat} u_A(t, \strat)$. As standard in the literature, the agent breaks ties in favor of the principal's utility. 
When both the agent and the principal are indifferent, ties are broken arbitrarily and consistently\footnote{Since the set of strategies is finite, one can impose a full order over the set, thus consistently is well defined.}.

We let $u_P(t)$ denote the principal's utility under the contract $t$ given that the agent best responds, i.e., $u_P(t) = u_P(t, \bestresponse(t))$.
Given a family of contracts $\cclass$, the principal's problem is to find a contract $t\in \cclass$ that maximizes $u_P(t)$.

\paragraph{Linear contracts.} 
A prevalent class of contracts is {\em linear contracts}, where the principal compensates the agent with a constant fraction of the reward.
A linear contract is given by a parameter $\alpha \in [0,1]$, where the payment is $t(j) = \alpha \cdot r(j)$. We denote a linear contract with parameter $\alpha$ by $t_\alpha$.

\paragraph{Approximation algorithms. } 
A $\beta$-approximation algorithm for the optimal contract returns a contract $t$ such that $u_P(t) \ge \beta \cdot u_P(t^\star)$, where $t^\star$ is an optimal contract.

\subsection{Pandora's Box Problem}
\label{sec:pandora}

A closely related problem is the well known Pandora's Box problem, introduced in \citep{weitzman1978optimal}. This problem models a search for the best alternative under uncertainty, when inspections are costly. In the Pandora's Box problem there is a set $\boxes$ of $n$ boxes. Each box $i\in \boxes$ is associated with an inspection cost $c_i \ge 0$ and a random value $V_i$. 
The variables $\{V_i\}_{i\in \boxes}$ are distributed according to a 
joint distribution $\mathcal{F}$.

At each point in time, a decision maker selects a box $i\in \boxes$ to open, incurring cost $c_i$, and revealing the value hidden inside (the value of the random variable $V_i$). The decision maker then decides whether to halt, in which case he obtains the maximum value revealed until this point, or continue and open another (yet-unopened) box. 

A strategy for a problem instance of Pandora's Box defines which box to open next and when to halt. For a given strategy $\strat$, let $S(\strat)$ denote the (random) set of boxes opened, and $V(\strat) = \max_{i\in S(\strat)}V_i$ denote the (random) maximum value observed by $\strat$, this is interpreted as $0$ if the strategy opens no boxes. The strategy's expected utility is given by
\[
u(\pi) = \E_{\{V_i\}_{i\in \boxes} \sim \mathcal{F}}[V(\pi) - \sum_{i\in S(\pi)}c_i].
\]

\paragraph{Relationship between the agent's problem and Pandora's Box.} The agent's problem in our sequential model, i.e., finding the best response $\bestresponse(t)$ to a contract $t$, 
relates
very naturally to an instance of the Pandora's Box problem. 
Indeed, given a contract $t$, to maximize his utility, the agent seeks to solve the Pandora's Box problem instance in which the set of boxes corresponds to the set of actions, the cost of box $i$ equals the cost of action $i$, and the value of box $i$ is $V_i = t(X_i)$, i.e., the payment for the outcome of the action associated with box $i$. We call this instance of the Pandora's Box problem the \emph{induced Pandora's Box instance}. Note that solving the induced Pandora's Box instance doesn't necessarily solve the problem of the agent's best response since, when more than one strategy maximizes the agent's utility, it may not be trivial which of them leads to the best principal's utility.

\section{Independent Actions} \label{sec:independent}
In this section, we consider the independent action model. 

We begin our analysis, in Section~\ref{sec:independent_best_response}, by characterizing the agent's optimal strategies (strategies which maximize his utility) under a given contract in the independent-action setting. This characterization gives insight into how these optimal strategies change when the contract changes. Additionally, we show that it is possible to compute the principal's utility from a contract in polynomial time.  

In Section~\ref{sec:opt_linear}, we present our first result: a polytime algorithm for the optimal linear contract.
In Section~\ref{sec:independent_general_positive} we consider the problem of computing the best general contract and show that this can be done in polynomial time in cases where $m$ is constant. In Section~\ref{sec:independent_general_negative} we show that without the restriction of $m$ being constant, the problem of computing the best general contract is NP-hard.

Additionally, in Appendix~\ref{app:independent_lin_opt_ratio} we quantify the worst-case ratio between the principal's utility under general contracts and linear contracts, and show that this ratio can be as high as $\Omega(n)$ but is upper bounded by $O(n^2 m)$.

\subsection{The Agent's Optimal Strategies} \label{sec:independent_best_response}
We call any strategy $\strat$ that maximizes the agent's utility in response to a contract $t$ an \emph{optimal strategy}. As noted in Section~\ref{sec:model}, the agent's optimal strategies are the optimal strategies in the induced Pandora's Box instance. \citet{weitzman1978optimal} gives a simple description of these strategies in the independent-action model.

\paragraph{Weitzman's algorithm.} 
\citet{weitzman1978optimal} gives a simple description of all optimal strategies in a Pandora's Box problem instance where the variables $\{V_i\}_{i\in \boxes}$ are independent (i.e., for every $i\in \boxes$, $V_i$ is sampled independently from a probability distribution $\mathcal{F}_i$). 
At the core of Weitzman's algorithm is the notion of a reservation value of box $i$, defined as the value $z_i$ satisfying
\begin{equation} \label{eq:res_value}
\E_{\{V_i\}_{i\in \boxes} \sim \mathcal{F}}[(V_i - z_i)^+] = c_i,
\end{equation}
where $(X)^+$ denotes $\max(X, 0)$. When $c_i = 0$, we define $z_i = \infty$. 

\citet{weitzman1978optimal} shows that any optimal strategy opens the boxes in non-increasing order of $z_i$, until the maximum value revealed so far\footnote{We interpret the maximum revealed value to be $0$ if no boxes have been opened yet.} exceeds (or, possibly, meets) the next reservation value, or until all boxes have been opened. Furthermore, any strategy satisfying these conditions is optimal.
More formally,
\begin{proposition} [\cite{weitzman1978optimal}] \label{prop:pandora_opt_strat}
    A strategy $\strat$ of the Pandora's Box problem is optimal if and only if the following conditions are always satisfied at each point in time: 
\begin{enumerate}
    \item When $\strat$ opens box $i\in \boxes$, it holds that $z_i \ge z_{i'}$ for any unopened box $i'\in \boxes$.
    \item Let $v$ be the maximum value revealed so far, and let $i\in \boxes$ be an unopened box of maximum $z_i$.  If $v > z_i$, $\strat$ halts, and if $v < z_i$, $\strat$ doesn't halt.
\end{enumerate}
\end{proposition}

This implies that an agent's strategy is optimal if and only if it satisfies the conditions of Proposition~\ref{prop:pandora_opt_strat} in the induced Pandora's Box instance and upon halting chooses a revealed outcome that maximizes the agent's payment. More formally, for any action $i\in \actions$, we denote by $z_i(t)$ the reservation value of box $i$ in the Pandora's Box instance induced by a contract $t$, and get the following:
\begin{observation} \label{obs:agent_opt_strat_cond} 
    For any contract $t$, an agent's strategy $\strat$ is optimal if and only if the following conditions are always satisfied at each point in time:
    \begin{enumerate}
        \item When $\strat$ takes an action $i\in \actions$ it holds that $z_i(t) \ge z_{i'}(t)$ for any untaken action $i'\in \actions$.
        \item Let $j\in \outcomes$ be a revealed outcome of maximum payment, and let $i\in \actions$ be an untaken action of maximum $z_i(t)$. If $t(j) > z_i(t)$, $\strat$ halts, and if $t(j) < z_i(t)$, $\strat$ doesn't halt.
        \item Upon halting, it holds that $t(\omega(\strat)) \ge t(j)$ for any revealed outcome $j$ (recall that $\omega(\strat)$ is the final outcome chosen by $\strat$, see Section~\ref{sec:model}).
    \end{enumerate}
\end{observation}

\paragraph{Computing the principal's utility.}
We now turn to the task of computing the principal's utility from a given contract. Specifically, we show the following proposition.
\begin{restatable}{proposition}{princutility}
\label{prop:principal_utility}
The principal's utility from a contract is computable in polynomial time.
\end{restatable}

There are a few challenges we face here. Firstly, the number of optimal strategies may be exponential, so finding an optimal strategy that also maximizes the principal's utility is non-trivial. Also,
despite the existence of very simple optimal strategies, some optimal strategies may be complex, and it is not clear if the principal's utility from a complex strategy can be computed efficiently. 

To address these challenges, for any contract $t$ and $\varepsilon>0$, we define the contract $\epsiloncontract$ as $\epsiloncontract(j) = t(j) + \varepsilon \cdot (r(j)-t(j))$. We then show that for a small enough $\varepsilon$, any optimal strategy under the contract $\epsiloncontract$, is also an optimal strategy under $t$. We continue to show that in such a case, an optimal strategy under the contract $\epsiloncontract$ also maximizes the principal's utility under $t$. 

We then restrict our attention to a simple type of agent's strategies which we call \emph{non-adaptive} strategies. In such strategies, the agent's choice of the order in which to take actions is fixed in advance, and his choice of when to halt depends only on the maximum revealed outcome so far. It is easy to see that non-adaptive optimal strategies exist (and are easy to find). The important property of non-adaptive strategies is that the principal's utility when the agent engages in one can be computed efficiently. 
Therefore, by finding a strategy $\strat$ that is a non-adaptive optimal strategy under the contract $t^\varepsilon$, we can compute $u_P(t)$ by computing $u_P(t, \strat)$.

For the full details, see Appendix~\ref{app:independent_principal_utility}.

\subsection{Optimal Linear Contracts}  \label{sec:opt_linear}
In this section, we prove that the optimal linear contract can be found in polynomial time. More formally, we prove the following theorem:

\begin{theorem} \label{thm:polytime_linear_contract}
    In the independent action model, the optimal linear contract can be computed in polynomial (in $n$ and $m$) time.
\end{theorem}

Our proof of Theorem~\ref{thm:polytime_linear_contract} relies on the concept of \emph{critical values}, introduced next.
\paragraph{Critical values. }
The concept of \emph{critical values} follows from
an analysis by \citet{dutting2019simple}. Put in terms of our problem, for any strategy $\strat$ we denote by $c(\strat) = \E[\sum_{i\in S(\strat)}c_i]$ the strategy's expected cost to the agent, and by $R(\strat) = \E[r(\finaloutcome(\strat))]$ the strategy's expected principal's reward. Under a linear contract $t_\alpha$ and a strategy $\strat$, the agent and principal's utilities are given by
\[
u_A(t_\alpha, \strat) = \alpha \cdot R(\strat) - c(\strat)~\text{and}~u_P(t_\alpha, \strat) = (1-\alpha) R(\strat)
\]

From this description of the principal's utility it follows that strategies that maximize $R(\strat)$ are the ones that maximize the principal's utility.
From the description of the agent's utility it follows that the agent's optimal strategies only changes when $\alpha = \frac{c(\strat_1) - c(\strat_2)}{R(\strat_1) - R(\strat_2)}$ for some two strategies $\strat_1, \strat_2$ with $R(\strat_1) \ne R(\strat_2)$. At such points of $\alpha$, the agent is indifferent between the strategies $\strat_1$ and $\strat_2$ (and possibly other strategies), and therefore breaks ties in favor of the highest expected reward, which, as one can verify, are the same strategies that maximize his utility to the right of $\alpha$. 
Due to consistent tie-breaking we get that agent's best response $\bestresponse(t_\alpha)$ only changes finitely many times and is right-continuous.

The values of $\alpha$ where the agent's best response changes are called the \emph{critical values}. 
An important note to stress that follows from this analysis is the following:
\begin{observation} \label{obs:crit_value_opt_strat} 
    If $\alpha$ is a critical value then $\bestresponse(t_\alpha)$ --- the agent's best response to the contract $t_\alpha$ --- wasn't an agent's optimal strategy under any linear contract $t_{\alpha'}$ where $\alpha' < \alpha$.
\end{observation}

It is easy to see that the optimal value of $\alpha$ is also a critical value since the principal seeks to minimize $\alpha$ while incentivizing the agent to engage in some strategy. 
By enumerating over the critical values and computing the principal's utility for each, we can choose the best among them and are guaranteed it maximizes the principal's utility.

Therefore, to prove Theorem~\ref{thm:polytime_linear_contract}, it suffices show that iterating over all critical values can be done in polynomial time (in $n$ and $m$). This is established in the following proposition.

\begin{proposition} \label{prop:polymany_critvalues}
    In the independent action model, there are $O(n^2 m)$ critical values of $\alpha$, and they are computable in polynomial time.
\end{proposition}

\begin{proof}
To prove Proposition~\ref{prop:polymany_critvalues} we characterize $z_i(\alpha)$ --- the reservation values as a function of $\alpha$ --- and show a relationship between $z_i(\alpha)$ and the critical values of $\alpha$. 

The characterization of $z_i(\alpha)$ is cast in the following lemma.
\begin{restatable}{lemma} {respiecwiselin}\label{lem:res_piecwise_lin}
    For every action $i\in \actions$, $z_i(\alpha)$ is a convex piecewise linear function with at most $m$ segments and is given by an efficiently computable closed-form expression. 
\end{restatable}

The proof of Lemma~\ref{lem:res_piecwise_lin} is deferred to the appendix (see Appendix~\ref{app:proof_lem_res_piecewise_lin}).

We now describe a relation between the values of $z_i(\alpha)$, and the critical points of $\alpha$.
\begin{restatable}{lemma} {critvaluecond}\label{lemma:crit_value_cond}
    If $\alpha'$ is a critical value, then one of the following holds:
    \begin{enumerate}
    \itemsep0em 
        \item For some two actions $i_1 \ne i_2$, it holds that $z_{i_1}(\alpha)$ and $z_{i_2}(\alpha)$ intersect at $\alpha'$.
        \item For some action $i \in \actions$ and outcome $j\in \outcomes$, it holds that $z_{i_1}(\alpha)$ intersects with $\alpha \cdot r(j)$ at $\alpha'$. 
    \end{enumerate}
\end{restatable}
Lemma~\ref{lemma:crit_value_cond} follows from an analysis of the conditions in Observation~\ref{obs:agent_opt_strat_cond}. The full details of the proof can be found in Appendix~\ref{app:proof_lemma_crit_value_cond}. 

Overall, we have that only points which that satisfy either condition of Lemma~\ref{lemma:crit_value_cond} may be critical values. 
Any two piecewise linear functions with at most $m$ segments intersect at most $m$ times, which, due to Lemma~\ref{lem:res_piecwise_lin}, implies that at most $O(n^2 m)$ points satisfy the first condition of Lemma~\ref{lemma:crit_value_cond}. Since $z_i(\alpha)$ is convex, it intersects with any straight line at most twice, which implies at most $O(n\cdot m)$ points satisfy the second condition of Lemma~\ref{lemma:crit_value_cond}. The values of $\alpha$ at which these intersections happen can be computed in $O(m)$ time from our explicit expression for $z_i(\alpha)$ (given by Lemma~\ref{lem:res_piecwise_lin}), and are the only candidates for critical values of $\alpha$.
The principal's utility from each such candidate can be computed in $O(n m^2)$ time (see Appendix~\ref{app:independent_principal_utility}), so overall our algorithm has a running time of $O((n^2 m+n m) (m + n m^2)) = O(n^3 m^3)$.
\end{proof}

\begin{remark} \label{remark:crit_points_m_dependence}
    In Appendix~\ref{app:crit_points_m_dependence}, we show that the dependence on $m$ is necessary, and in fact even with only $n=2$ actions there may be $\Omega(m)$ critical points. This presents a significant departure from the simultaneous model, in which under linear contracts the number of outcomes has no effect on the number of critical points.
\end{remark}

\subsection{Computing Optimal General Contracts When $m$ is Constant}\label{sec:independent_general_positive}
In this section, we prove that when $m$ is constant, the best general
contract can also be found in polynomial time. This is cast in the following theorem.

\begin{theorem} \label{thm:opt_general_contract}
    In the independent action model, the optimal general contract can be computed in polynomial time, in cases where $m$ is constant. 
\end{theorem}

Before turning to the proof of Theorem~\ref{thm:opt_general_contract}, we set up a framework. Firstly, we expand our analysis from Section~\ref{sec:independent_best_response} and examine more closely how the agent's optimal strategies change as a function of the contract. Then, as our proof of Theorem~\ref{thm:opt_general_contract} relies heavily on hyperplane arrangements, we provide some preliminaries regarding these.
\paragraph{The set of optimal strategies. }
Based on the description of the agent's optimal strategies given in Observation~\ref{obs:agent_opt_strat_cond}, all optimal strategies under some contract $t$ are determined by a 4-tuple $(\actionspreorder, \outcomespreorder, \haltgreater, \haltsmaller)$ where 
\begin{itemize}
    \item $\actionspreorder$ is the full preorder\footnote{A preorder $P$ is a reflexive and transitive binary relation. A full preorder also satisfies $x \le_P y$ or $ y \le_P x$ for any $x,y$.} induced on the set of actions $\actions$ according to $i_1 \le_{\actionspreorder} i_2$ if and only if $z_{i_1}(t) \le z_{i_2}(t)$. $\actionspreorder$ determines the possible action-taking order of an optimal strategy.
    \item $\outcomespreorder$ is the full preorder induced on the set of outcomes $\outcomes$ according to $j_1 \le_{\outcomespreorder} j_2$ if and only if $t(j_1) \le t(j_2)$. $\outcomespreorder$ determines which outcomes an optimal strategy may choose upon halting.
    \item $\haltgreater = \{\haltgreater_i\}_{i\in \actions}\subseteq 2^{\outcomes}$ and $\haltsmaller = \{\haltsmaller_i\}_{i\in \actions}\subseteq 2^{\outcomes}$ are defined by 
    \[
    \haltgreater_i = \{j \in \outcomes \mid t(j) > z_i(t)\} ~\text{and}~\haltsmaller_i = \{j \in \outcomes \mid t(j) < z_i(t)\}.
    \]
    $\haltgreater$ and $\haltsmaller$ determine when an optimal strategy must halt and when it must continue.
\end{itemize}
We sometimes use $P_\actions(t), P_{\outcomes}(t), \haltgreater(t), \haltsmaller(t)$ denote the values of $P_\actions, P_{\outcomes}, \haltgreater, \haltsmaller$ under a contract $t$,
but usually we neglect the contract which defines them for brevity.

\paragraph{Hyperplane arrangements.} A \emph{hyperplane arrangement} is a set $\mathcal{A}$ of hyperplanes in some vector space. We focus on hyperplanes over the reals. 
Consider some hyperplane arrangement $\mathcal{A}$ over $\reals^d$.
When the hyperplanes of $\mathcal{A}$ are removed, the space is divided into components, each one being an open $d$-dimensional polytope. These are called the \emph{regions} of $\mathcal{A}$. 

A \emph{face} of an arrangement $\mathcal{A}$ is an intersection of the closure of a region with some intersection of the hyperplanes in $\mathcal{A}$, i.e,
\[
\bar{R} \cap \bigcap_{A\in \mathcal{I}}A,
\]
where $R$ is a region of $\mathcal{A}$ and $\mathcal{I} \subseteq \mathcal{A}$. 

The \emph{dimension} of a face $F$ is defined as the dimension of the affine span of $F$. A \emph{$k$-face} of $\mathcal{A}$ is a face of dimension $k$.

Two useful definitions in this context are  the \emph{relative interior} of a set $S$, which we denote by $\relint(S)$ and is defined as the interior of $S$ within the affine span of $S$, and the \emph{relative boundary} of $S$, which is its boundary within the affine space. An important property is that, for any $k>0$, the relative boundary of any $k$-face is comprised of $k-1$-faces.
\begin{observation} \label{obs:face_relint_cond}
    For any face $F$ of $\mathcal{A}$ and any hyperplane $A\in \mathcal{A}$, either $\relint(F) \subseteq A$ or $\relint(F) \cap A = \emptyset$.
\end{observation}
Finally, we rely on the following proposition:
\begin{proposition} [\cite{zaslavsky1975facing}]
    The number of $k$-faces in a hyperplane arrangement of size $n$ in $\reals^d$ is at most
    \[
    {n \choose d-k} \sum_{i=0}^k {n-d+k \choose i}.
    \]
\end{proposition}
This implies that the number of faces of an arrangement $\mathcal{A}$ is bounded by $O(|\mathcal{A}|^d)$.

\begin{proof} [Proof of Theorem~\ref{thm:opt_general_contract}]
Denote $L = \max_{j\in \outcomes, i\in \actions: p_{ij}\ne 0} \frac{r(m)}{p_{ij}}$. It is easy to verify that any optimal contract $t$ satisfies $t(j)\in [0,L]$ for any outcome $j\in \outcomes$.
For our proof we think of a contract $t$ as an element of $\contractregion$, where the $j$th coordinate of the vector is the payment for outcome $j$. 
The basic idea is to construct a hyperplane arrangement which captures changes in the 4-tuple $(\outcomespreorder, \actionspreorder, \haltgreater, \haltsmaller)$. The intuition here is that, in the same way that under linear contracts the points where the agent's best response changes are the candidates for the optimal contracts, in a general contract the points where the 4-tuple changes are of interest. More formally, if within the relative interior of each face of a hyperplane arrangement the 4-tuple $(\outcomespreorder, \actionspreorder, \haltgreater, \haltsmaller)$ is constant, we show that the optimal contract is a $0$-face (a point which is the intersection of $m$ hyperplanes).

We begin by characterizing the reservation value as piecewise linear. This is essential in order to capture changes in $\actionspreorder$ with hyperplanes.

\begin{lemma} \label{lem:linz}
    For any $i\in \actions$, there exist $2^{m}$ linear functions, denoted by $\linz{S}_i(t)$ for any $S\in 2^{\outcomes}$, such that for any contract $t$ it holds that $z_i(t) = \linz{S}_i(t)$, where $S=\haltgreater_i(t)$.
\end{lemma}
\begin{proof}
    Let $i\in \actions$. For any subset $S\subseteq \outcomes$ we define $\linz{S}_i(t)$ as follows:
    \[
    \linz{S}_i(t) = \frac{\sum_{j\in S} p_{ij}\cdot t (j) - c_i}{\sum_{j\in S} p_{ij}},
    \]
    Let $t$ be some contract.
    The random variable $t(X_i)$ attains values within $\{t(j)\}_{j\in \outcomes}$ and is equal to $t(j)$ with probability $p_{ij}$. Plugging this into Equation~\eqref{eq:res_value} we get
    \[
    c_i = \E[(t(X_i)-z_i(t))^+] = \sum_{j\in \outcomes: t(j) > z_i(t)} p_{ij} (t(j)-z_i(t))  = \sum_{j\in \haltgreater_i(t)} p_{ij}(t(j) - z_i(t)), 
    \]
    where the last equality is by the definition of $\haltgreater_i(t)$. This simplifies into
    \[
    z_i(t) = \frac{\sum_{j\in \haltgreater_i(t)} p_{ij}\cdot t (j) - c_i}{\sum_{j\in\haltgreater_i(t)} p_{ij}} = \linz{\haltgreater_i(t)}_i(t),
    \]
    as needed.
\end{proof}

We now begin to formalize our earlier notion of the candidates for the optimal contract being $0$-faces of an arrangement. We begin our analysis by considering the principal's utility in polytopes (which, in a sense, are the basic building block of hyperplane arrangements). 
\begin{restatable}{lemma}{polytopemax}\label{lem:polytope_max}
    Let $Q \subseteq \contractregion$ be a polytope. $u_P(t)$ attains a maximum\footnote{The principal's utility is not continuous therefore this is not immediate from the Extreme Value Theorem.} on $Q$.
\end{restatable}
The proof of this lemma is deferred to Appendix~\ref{app:proof_lemma_polytope_max}.

Furthermore, when the agent's optimal strategies remain the same throughout the relative interior of a polytope, this maximum is attained at the boundary. More formally,
\begin{restatable}{lemma}{polytopeboundary} \label{lem:polytope_boundary}
    Let $Q\subseteq \contractregion$ be a polytope such that in the relative interior of $Q$ the 4-tuple $(\actionspreorder, \outcomespreorder, \haltgreater, \haltsmaller)$ is constant. The principal's utility is maximized at the relative boundary of $Q$.
\end{restatable}
The proof of this lemma relies on two observations. Firstly, that for any single strategy $\strat$ the principal's utility $u_P(t, \strat)$ is maximized at a vertex. Secondly, any strategy which maximizes the agent's utility at the relative interior of a polytope also maximizes it at the boundary. The full details are deferred to Appendix~\ref{app:proof_lemma_polytope_boundary}.
As a corollary, we get the following key finding.
\begin{corollary} \label{cor:0_faces}
    Let $\mathcal{A}$ be a hyperplane arrangement, such that any face $F$ of $\mathcal{A}$ satisfies either $F\subseteq \contractregion$ or $\relint(F) \cap \contractregion = \emptyset$. If any face $F$ such that $F\subseteq \contractregion$ satisfies the condition of Lemma~\ref{lem:polytope_boundary}, then some $0$-face (a point which is the intersection of $m$ hyperplanes) of $\mathcal{A}$ maximizes the principal's utility.
\end{corollary}
To prove Theorem~\ref{thm:opt_general_contract}, it suffices to show a hyperplane arrangement of polynomial size that can be computed in polynomial time and satisfies the conditions of the Corollary~\ref{cor:0_faces}. This is because, in a vector space of constant dimension, all $0$-faces of a hyperplane arrangement can be found in polytime. 

Consider the following hyperplane arrangements: 
\begin{enumerate}
\item $\mathcal{A}_1 = \{\{t\in \reals^m \mid t_j = 0\} \mid j\in [m]\} \cup \{\{t\in \reals^m \mid t_j = L\} \mid j\in [m]\}$
\item $\mathcal{A}_2 = \{E_{j_1, j_2} \mid  j_1\ne j_2 \in [m]\}$, where $E_{j_1, j_2}$ is defined as 
\[
E_{j_1, j_2}=\{t\in \reals^m \mid t_{j_1} = t_{j_2} \}.
\]
\item $\mathcal{A}_3 = \{ T_{\rho, i,j} \mid \rho:[m]\rightarrow [m]\text{ is a bijection}, i\in \actions, j\in [m]\}$, where $T_{\rho, i,j}$ is defined as \[
T_{\rho, i,j} = \{t\in \reals^m \mid \sum_{k=j}^m p_{ik}(t(\rho(k))-t(\rho(j))) = c_i\}. 
\]
\item $\mathcal{A}_4 = \{Z_{i_1, i_2, S_1, S_2} \mid i_1\ne i_2 \in \actions, S_1, S_2 \subseteq \outcomes\}$
where $Z_{i_1, i_2, S_1, S_2}$ is defined as 
\[
Z_{i_1, i_2, S_1, S_2} = \{t\in \reals^m \mid \linz{S_1}_{i_1}(t) = \linz{S_2}_{i_2}(t)\}.
\]
\end{enumerate}
We consider the hyperplane arrangement $\mathcal{A} = \mathcal{A}_1 \cup \mathcal{A}_2 \cup \mathcal{A}_3 \cup \mathcal{A}_4$, and show that it satisfies the conditions of Corollary~\ref{cor:0_faces}.

It is immediate that $\mathcal{A}_1$ guarantees the condition that any face $F$ satisfies $F \subseteq \contractregion$ or $\relint(F)\cap \contractregion = \emptyset$, i.e., it creates a separation between regions which define a contract ($\contractregion$) and regions which do not. The hyperplanes of $\mathcal{A}_2$ capture the transition points in $P_{\outcomes}$, these are points at which the agent's preference between outcomes might change. More specifically, each hyperplane $E_{j_1, j_2}$ captures the places where the agent may change his preference between $j_1$ and $j_2$. 

The hyperplanes of $\mathcal{A}_3$ capture the transition points in $\haltgreater$ and $\haltsmaller$. More specifically for any bijection $\rho: \outcomes \rightarrow \outcomes$, within regions that satisfy $t(\rho(1)) \le \dots \le t(\rho(m))$, each hyperplane $T_{\rho, i, j}$ captures the points where the agent might change his choice of halting before action $i$ when the highest outcome revealed is $\rho(j)$.

Finally, the hyperplanes of $\mathcal{A}_4$ capture transition points in $\actionspreorder$. More specifically, for any two actions $i_1\ne i_2\in \actions$, within regions that satisfy $\haltgreater_{i_1} = S_1$ and $\haltgreater_{i_2} = S_2$, the hyperplane $Z_{i_1, i_2, S_1, S_2}$ captures the points where the agent might choose to switch his ordering of $i_1$ and $i_2$.

\begin{restatable}{proposition}{finalhyperplanearrangment}\label{prop:final_hyperplane_arrangment}
    For any face $F$ of the hyperplane arrangement $\mathcal{A}$ such that $F\subseteq \contractregion$, it holds that, within $\relint(F)$, the 4-tuple $(\actionspreorder, \outcomespreorder, \haltgreater, \haltsmaller)$ is constant.
\end{restatable}
The proof of this proposition relies on an analysis of each of the tuple's components, together with Observation~\ref{obs:face_relint_cond}, and is deferred to Appendix~\ref{app:hyperplane_arrangement_correctness}.

To conclude the proof we note that $|\mathcal{A}_1|=m$, $|\mathcal{A}_2| \le m^2$, $\mathcal{A}_3 \le m! \cdot n \cdot m$, and $|\mathcal{A}_4| \le n^2 \cdot 2^{2m}$, so overall $|\mathcal{A}|$ is polynomial in $n$ as needed. It is also clear that the hyperplanes of $\mathcal{A}$ can be computed in polynomial time. This, together with Proposition~\ref{prop:final_hyperplane_arrangment} and Corollary~\ref{cor:0_faces} concludes the proof of Theorem~\ref{thm:opt_general_contract} by showing an algorithm with a running time of $O(|\mathcal{A}|^m (n m^2+m^3))=O(n^{2m+1})$, where $n m^2$ is the complexity of evaluating the principal's utility at each vertex and $m^3$ is the complexity of finding each vertex (equivalently, solving a system of $m$ linear equations in $m$ variables).
\end{proof}

\subsection{NP-hardness of Optimal General Contract for Non-Constant $m$}\label{sec:independent_general_negative}
In this section we complement our polytime algorithm from Section~\ref{sec:independent_general_positive} by showing that when removing the restriction of a constant $m$, computing the optimal general contract is NP-hard.
\begin{theorem} \label{thm:independent_np_hard}
    In the independent action model, the optimal general contract is NP-hard to compute
\end{theorem}
\begin{proof}
Our proof is via a reduction from PARTITION.
Let $a_1, \dots, a_k \in (0,1)$ such that $\sum_{i=1}^k a_i = 0.2$. We wish to decide whether there exists a subset $S\subseteq [k]$ such that $\sum_{i\in S} a_i =0.1$. 
Consider the following contract instance with $k+2$ outcomes and $3$ actions (and parameters $\varepsilon= 0.01 \min_{i\in [k]}a_i$, $q$ that solves $q^2 (-10 + \varepsilon)+q (-2 - 0.8 \varepsilon)  +0.9- 0.99 \varepsilon=0$, which is approximately $q\approx 0.216$, and $c = \frac{\varepsilon q}{10}$): 
\begin{enumerate}
    \item Here outcome $0$ (instead of outcome $1$) is the zero outcome which is always available to the agent (even without taking any actions).
    \item $r_0 = r_1 = \dots = r_{k} = 0$, and $r_{k+1} = 1$.
    \item Action $1$ is free ($c_1= 0$) and leads to outcome $k+1$ with probability $1-\varepsilon$ and otherwise to outcome $0$ (i.e., $p_{1, k+1} = 1-\varepsilon$ and $p_{1, 0} = \varepsilon$).
    \item Action $2$ is free ($c_2 = 0$) and leads to each outcome $j\in [k]$ with probability $a_j$, otherwise it leads to outcome $0$ (i.e., $p_{2,j} = a_j$ for each $j\in [k]$ and $p_{2, 0} = 1- \sum_{i=1}^k a_i = 0.8$).
    \item Action $3$ has cost $c$ ($c_3 = c$) and leads to outcome $k+1$ with probability $q$, leads to each outcome $j \in [k]$ with probability $a_j$, and otherwise leads to outcome $0$ (i.e.,  for any $j\in [k]$, $p_{3, j} = a_j$ and $p_{3, k+1} = q$, $p_{3, 0} = 1-\sum_{i=1}^k a_i-q = 0.8-q$).
\end{enumerate}

This construction is given in Table~\ref{table:np_hardness_construction}.

\begin{table}[h!]
\centering
\begin{tabular}{ |c|c|c|c|c|c|c| } 
 \hline
 Action & $c_i$ & $p_{i, 0}$ & $p_{i, 1}$ & \dots & $p_{i,k}$ & $p_{i, k+1}$ \\
 \hline
 Action $1$ & $0$ & $\varepsilon$ & $0$ & \dots & $0$ & $1-\varepsilon$\\ 
 Action $2$ & $0$ & $0.8$ & $a_1$ & \dots & $a_k$ & $0$\\ 
 Action $3$ & $c$ & $0.8-q$ & $a_1$ & \dots & $a_k$ & $q$\\ 
 \hline
\end{tabular}
\caption{An overview of costs and probabilities in the reduced sequential contracts instance.}\label{table:np_hardness_construction}
\end{table}

Before proceeding with the proof, we give intuition. The idea here is that the agent's best response is always to start by taking actions $1$ and $2$, then potentially also take action $3$, and finally choose an outcome with maximal $t(j)$. Thus, the contract design problem is essentially how to best incentivize the agent to take action $3$ (the only costly action), while ensuring the agent also selects outcome $k+1$ (the only good outcome). A crucial feature of this construction is that the optimal contract $t^\star$ always satisfies $t^\star(k+1) \ge t^\star(i)$ for any $i\in [k]$; intuitively, this is to ensure that the agent chooses the good outcome when halting. This key property, along with the fact that the principal wants to keep $t^\star(k+1)$ as low as possible, since this will be the payment to the agent in most cases, implies that an optimal contract is of a very special form, fully described by a subset $S\subseteq [k]$ (see Definition~\ref{def:equal_spread_contract}). Furthermore, the principal's utility is maximized if and only if this subset satisfies $\sum_{j\in S} a_j = 0.1$.

The following observation is immediate from our construction and from the optimality of Weitzman's algorithm (see Observation~\ref{obs:agent_opt_strat_cond})
\begin{observation} \label{obs:agent_best_response}
    The agent's best response to any contract $t$ is to take actions $1$ and $2$, then take action $3$ if and only if $z_3(t) \ge 0$ and no outcome $j$ with $t(j) > z_3(t)$ has been revealed, and then choose an outcome with maximal $t(j)$ as the final outcome.
\end{observation}

We now define a simple form of contracts which is optimal for our constructed problem instance.
\begin{definition} \label{def:equal_spread_contract}[equal-spread contract]
    In a reduced contract instance, an \emph{equal-spread contract} with respect to a set $S\subseteq[k]$ equals
    \[
    t^S(j) = \frac{c}{q+\sum_{\ell\in S} a_\ell} \cdot \indicator[j\in S \cup \{k+1\}]
    \]
\end{definition}

The equal-spread contract $t^S$ incentivizes the agent to take action $3$ if and only if action $1$ yielded outcome $0$ and action $2$ yielded an outcome not in $S$.

\begin{restatable}{claim}{equalspreadutil}\label{claim:equal_spread_util}
    For any $S\subseteq [k]$, the principal's utility from the equal-spread contract $t^S$ is given by 
    \[
    u_P(t^S) = \underbrace{1-\varepsilon}_{\mathclap{\text{exp. reward from action $1$}}} + \overbrace{\varepsilon(1- x) q}^{\mathclap{\text{exp. reward from action $3$}}} - \underbrace{\left(1- \varepsilon (1- x) (1-x-q)\right)}_{\mathclap{\text{prob. of positive payment}}} \overbrace{\frac{c}{x+q}}^{\mathclap{\text{positive payment amount under $t^S$}}},
    \]
    where $x=\sum_{j\in S} a_j$.
\end{restatable}
The proof of Claim~\ref{claim:equal_spread_util} is deferred to Appendix~\ref{app:equal_spread_util}.

\begin{restatable}{claim}{optisequalspread} \label{claim:opt_is_equal_spread}
    The optimal contract is an equal-spread contract.
\end{restatable}
We prove  Claim~\ref{claim:opt_is_equal_spread} by leveraging Claim~\ref{claim:equal_spread_util} and Observation~\ref{obs:agent_best_response} to
show that an optimal contract $t^\star$ equals to the equal-spread contract $t^S$, where $S= \{j\in [k] \mid t^\star(j) > z_{3}(t^\star)\}$. The proof is deferred to Appendix~\ref{app:opt_is_equal_spread}.

Our choice of $\varepsilon,q,c$ is such that $x=0.1$ uniquely maximizes the expression given for the principal's utility by Claim~\ref{claim:equal_spread_util}, in the range $x\in [0, \infty)$. Additionally, when $x=0.1$ the principal's utility is strictly greater than $1-\varepsilon$.
This, combined with Claim~\ref{claim:opt_is_equal_spread} and Claim~\ref{claim:equal_spread_util} implies that the answer to the original PARTITION problem instance is ``yes'' if and only if $t^\star(k+1)=\frac{c}{s+q}$, thus implying the NP-hardness of computing an optimal contract.
\end{proof}

\section{Correlated Actions} \label{sec:corr} 
In this section we establish a hardness result for sequential contract design in the correlated action model. 
The complexity of algorithms in this model is given with respect to the size of the support of the multi-dimensional distribution from which the random variables $\{X_i\}_{i\in \actions}$ are drawn. 

Our main result in this section is that in the correlated action model, it is NP-hard to approximate the optimal contract within any constant, even in the binary-outcome model (see Theorem~\ref{thm:corr_contract_hardness}). 
This result also implies a hardness result for the optimal linear contract in this setting, since in a binary-outcome model, there exists an optimal contract that is linear.

The remainder of this section is structured as follows. 
In Section~\ref{sec:init_corr} we formally define the binary-outcome model and present necessary preliminaries. 
In particular, we show an interesting equivalence between coverage set functions (see Definition~\ref{def:coverage}) and ``correlated OR'' set functions (see Proposition~\ref{prop:coverage_equiv}).
In Section~\ref{sec:corr_contract} we present our hardness result for correlated-action sequential contracts.

\subsection{Correlated Actions: Preliminaries} \label{sec:init_corr}
We restrict attention to the binary-outcome scenario. Since our result is a hardness one, this only strengthens it. 
In the binary-outcome scenario, the set of outcomes is $\{1, 2\}$, and the corresponding rewards are $r(1) = 0$ and $r(2)=1$. 

In the Pandora's Box problem corresponding to the agent's best response problem, this means that all random variables $\{V_i\}_{i\in \boxes}$ attain values in $\{0, r\}$, for some $r\in (0,1]$. We refer to this problem as the binary-outcome correlated Pandora's Box problem.

The following proposition by \citet{chawla2020pandora} shows that an optimal strategy in the binary-outcome correlated Pandora's Box problem takes a simple form:
\begin{proposition} \citep{chawla2020pandora}\label{prop:strat_corr_pandora_binary}
In the binary-outcome correlated Pandora's Box problem any optimal strategy can be represented as a tuple $\strat = (i_1, \dots, i_\ell)\in \boxes^\ell$ (for some $\ell \ge 0$). This tuple
is interpreted as opening boxes $i_1,\dots, i_\ell$ in this order, 
halting when a non-zero value is revealed or when all boxes in $\strat$ have been opened.
\end{proposition}
We assume without loss of generality that in such representations 
the last box is opened with a non-zero probability.
Proposition~\ref{prop:strat_corr_pandora_binary} implies that, in a binary-outcome correlated-action model, the agent's and principal's utilities
can be fully described by the set of actions $A$, the costs $\{c_i\}_{i \in A}$, and the ``correlated OR'' set function $f:2^A \rightarrow [0,1]$, where 
\begin{equation}
\label{eq:correlated-or}
f(S) = Pr[\exists i\in S.~X_i > 0].
\end{equation}
Indeed, the agent's expected cost from a strategy $\strat=(i_1,\ldots,i_\ell)$ is given by 
\[
c(\strat) = \sum_{k=1}^\ell (1-f(\{i_1, \dots, i_{k-1}\})) c_k.
\]
Then, given a contract $t_\alpha$ for some $\alpha \in [0,1]$ (in the binary case, linear contracts are without loss of generality), the agent's and principal's expected utilities are
\[
u_A(t_\alpha, \strat) = \alpha \cdot f(\{i_1, \dots, i_\ell\})  - c(\strat) ~\text{and}~u_P(t_\alpha, \strat) = (1-\alpha) f(\{i_1, \dots, i_\ell\}).
\]

This implies that a problem instance 
of the correlated-action sequential contracts 
can be represented by a tuple $(\actions, \{c_i\}_{i\in \actions}, f)$.
The ``correlated OR'' function $f$ is given by the OR of correlated Bernoulli random variables.

\begin{definition} \label{def:coverage}
A set function $f:2^A\rightarrow \reals_{\ge 0}$ is \emph{coverage} if there is a set of elements $U$, with associated non-negative weights  $\{w_u\}_{u\in U}$, and a mapping $h:A \rightarrow 2^U$ such that for every $S\subseteq A$, 
$f(S)=\sum_{u\in U } w_u \cdot \indicator[\exists i \in S.~u \in  h(i)] $.
A coverage function can be described by a tuple $(U,\{w_u \}_{u\in U}, A,h)$.
\end{definition}

The following proposition shows an equivalence between correlated OR functions (as in Equation~\eqref{eq:correlated-or}) and coverage functions. Importantly, it also shows that the coverage function $f$ has the same (asymptotic) representation size as that of the joint distribution of $\{X_i\}_{i\in \actions}$ when given explicitly (which we defined as the size of the support).  This means our choice of representation doesn't affect an algorithm's complexity as a function of input size.

\begin{proposition} \label{prop:coverage_equiv}
For every (correlated) Bernoulli random variables $\{X_i\}_{i\in A}$ with support $\mathcal{R} \subseteq \{0,1\}^A$ and probability density function $p: \mathcal{R} \rightarrow [0,1]$, there exists a coverage function  $f=(U, \{w_u\}_{u\in U}, A, h)$ with $|U| \le |\mathcal{R}|$ such that $f(S) =  Pr[\exists i\in S.~X_i > 0]$. 

Moreover, for every coverage function $f=(U, \{w_u\}_{u\in U}, A, h)$ such that $Im(f) \subseteq [0,1]$, there exist (correlated) Bernoulli random variables $\{X_i\}_{i\in A}$ with support $\mathcal{R}\subseteq \{0,1\}^A$ for which
$f(S) =  Pr[\exists i\in S.~X_i > 0]$,
and $|\mathcal{R}| \le |U|+1$.
\end{proposition}
\begin{proof}
    Let $\{X_i\}_{i\in A}$ be (correlated) Bernoulli random variables with support $\mathcal{R}$ and probability density function $p:\mathcal{R}\rightarrow [0,1]$. We consider the coverage function $f=(U, \{w_u\}_{u\in U}, A, h)$, where we define $U=\mathcal{R}$, for each $u\in U$ we define $w_u = p(u)$, and for each $i\in A$ we define $h(i)= \{u \in U \mid u_i > 0\}$. This yields
    \[
    Pr[\exists i\in S. ~X_i > 0] = \sum_{v\in \mathcal{R}} p(v)\indicator[\exists i\in S. ~v_i > 0] 
    = \sum_{u\in U } w_u \cdot \indicator[\exists i \in S.~u \in  h(i)] = f(S),
    \]
    as needed.

    In the other direction, let $f=(U, \{w_u\}_{u\in U}, A, h)$ such that $Im(f) \subseteq [0,1]$. Assume without loss of generality\footnote{This is without loss of generality because the only relevant elements of $U$ are those that belong to some $h(i)$, and the sum of their weights is $f(A) \le 1$.} that $\sum_{u \in U} w_u \le 1$.
    
    We define our variables $\{X_i\}_{i\in A}$ by sampling an element $x\in U \cup \{0\}$, where each element $u\in U$ is sampled with probability $Pr[x=u] = w_u$,
    and $0$ is sampled with probability $1-\sum_{u\in U} w_u$.
    For any $i\in A$ and $u\in U$, we denote $X_i(u) = \indicator[u\in h(i)]$ and $X_i(0) = 0$. Our random variables $\{X_i\}_{i\in A}$ are defined as $X_i = X_i(x)$, for any $i\in A$. Now:
    \[
    f(S) = \sum_{u\in U} w_u \cdot \indicator[\exists i \in S.~u \in  h(i)]  = \sum_{u\in U} Pr_x[x=u] \cdot \indicator[\exists i \in S. ~X_i(u) > 0] = Pr[\exists i\in S.~X_i > 0],
    \]
    as needed. To show the bound on the size of $\mathcal{R}$, we note that each $X_i$ is a function of $x$, and there are at most $|U|+1$ possible values of $x$ we get that $\mathcal{R}$, which is defined as support of $\{X_i\}_{i\in A}$, has size $|\mathcal{R}|\le |U|+1$.
\end{proof}

\begin{remark}
Proposition~\ref{prop:coverage_equiv} can be generalized further to an equivalence between coverage functions and ``correlated maximum'' functions, i.e., functions which map subsets of random variables to their expected maximum. This generalization is shown in Appendix~\ref{app:generalized_coverage_equiv}.
\end{remark}

It is known that any coverage function is also \emph{submodular}. A set function $f:2^A\rightarrow \reals_{\ge 0}$ is \emph{submodular} if for any two sets $S, S'\subseteq A$ s.t. $S\subseteq S'$ and $i\in A$ it holds that $f(i\mid S) \ge f(i \mid S')$, where $f(i \mid S) = f(S \cup \{i\}) - f(S)$ is the marginal contribution of $i$ to $S$.

It is also known that any submodular function (and therefore, any coverage function) is also \emph{subadditive}, which means for any $S,T\subseteq A$, $f(S\cup T) \le f(S) + f(T)$. 

\subsection{Hardness of Approximation of the Optimal Contract Under Correlated Actions} \label{sec:corr_contract}

\citet{chawla2020pandora} established a hardness of approximation for the correlated Pandora's Box problem. 
However, this hardness result does not directly imply a hardness of approximation for the optimal contract problem.
This is because the hardness of Pandora's Box corresponds to the hardness of the agent's best response problem, but even if the agent's problem is hard to approximate, it may still be the case that approximating the optimal contract admits a poly-time algorithm. 

In this section we present the following hardness of approximation result.
\begin{theorem} \label{thm:corr_contract_hardness}
    In the correlated action model, it is NP-hard to approximate the optimal contract within any constant, even in the binary-outcome case.
\end{theorem}

To prove this theorem, we construct a reduction from an NP-hard promise problem dealing with coverage functions, introduced in \citep{ezra2023approximability}. This promise problem is described in the following proposition:
\begin{proposition} [\citep{ezra2023approximability}] \label{prop:coverage_hardness} 
For every $M > 1$ and every $0 < \varepsilon < e^{-1}$, on input $(k,f)$, where $k\in \mathbb{N}$ and $f=(U, \{w_u\}_{u\in U}, A, h)$ is a coverage function with $Im(f)\subseteq [0,1]$, such that $\forall i\in A.~f(\{i\}) = \frac{1}{k}$ and exactly one of the following two conditions holds:
    \begin{enumerate}
        \item There exists a set $S\subseteq A$ of size $k$ such that $f(S) = 1$.
        \item Every set $S\subseteq A$ of size $\beta k$ such that $\beta \le M$ satisfies $f(S) \le 1-e^{-\beta}+\varepsilon$.
    \end{enumerate}
    It is NP-hard to determine which of the two conditions is satisfied by the input.    
\end{proposition}

\begin{proof} [Proof of Theorem~\ref{thm:corr_contract_hardness}]
Let $\gamma \in (0,1)$ be a constant, in what follows we rule out a polytime $\gamma$-approximation algorithm for the optimal contract via a reduction from the promise problem in Proposition~\ref{prop:coverage_hardness}.

For any $(k,f': 2^{A'} \rightarrow [0,1])$ per Proposition~\ref{prop:coverage_hardness} with the parameters $M=3, \varepsilon=0.001$, where $f' = (U, \{w_u\}_{u\in U}, A', h')$.
Construct the following instance of the sequential contracts problem: $(\actions, \{c_i\}_{i\in \actions}, f)$ defined as follows. Our set of actions is $\actions = A' \cup \{0\}$. Our costs are $c_i = \frac{1.5}{k+1}$ for any $i\in A'$ and $c_0 = 1-\frac{\gamma}{8}$. Our ``correlated OR'' function is $f=(U, \{w_u\}_{u\in U}, A, h)$, where $h$ is defined as 
\[
h(i) = \begin{cases}
    h'(i) & i\ne 0\\
    U & i=0.
\end{cases}
\]
Note that for any $S\subseteq A$ that doesn't contain $0$ we have $f(S) = f'(S)$, while for any $S\subseteq A$ that contains $0$ we have $f(S) = 1$.

\begin{restatable}{lemma} {reductioncorrectness} \label{lemma: reduction correctness}
    In the constructed sequential contracts problem instance, it holds that \begin{enumerate}
        \item When $f'$ satisfies condition (1) from Proposition~\ref{prop:coverage_hardness} the optimal principal's utility is at least $\frac{1}{4}$.
        \item When $f'$ satisfies condition (2) from Proposition~\ref{prop:coverage_hardness} it holds that $u_P(t_{1-\frac{\gamma}{8}}) > 0$ and $u_P(t_\alpha)=0$ for any $\alpha < 1-\frac{\gamma}{8}$.
    \end{enumerate}
\end{restatable}
The proof of Lemma~\ref{lemma: reduction correctness} is by observing that when condition (1) holds the contract $t_{\frac{3}{4}}$ incentivizes the agent to take all the actions in $A'$ until one of the actions is successful, thus yielding a principal's utility of $\left(1-\frac{3}{4}\right)\cdot 1 = \frac{1}{4}$. Under condition (2), we show that any contract $t_\alpha$ with $\alpha < \frac{\gamma}{8}$ incentivizes the agent to take no actions, thus yielding a utility of $0$, but under the contract $t_{1-\frac{\gamma}{8}}$ the agent takes the action $0$, yielding a utility of $(1-(1-\frac{\gamma}{8}))\cdot 1=\frac{\gamma}{8} > 0$. The details are deferred to Appendix~\ref{app:correlated}.

Let $t_\alpha$ be a $\gamma$-approximation of the optimal contract.
According to Lemma~\ref{lemma: reduction correctness}, under the first case of Proposition~\ref{prop:coverage_hardness} it holds that $\alpha \le 1-\frac{\gamma}{4}$ (since the principal's utility is upper-bounded by $1-\alpha$, and $u_P(t_\alpha) \ge \gamma \cdot \frac{1}{4}$), while under the second case it holds that $\alpha \ge 1-\frac{\gamma}{8}$. Thus $t_\alpha$ can be used to distinguish between the two cases of Proposition~\ref{prop:coverage_hardness}, which rules out a polytime $\gamma$-approximation algorithm (unless P=NP).

\end{proof}

\bibliographystyle{abbrvnat}
\bibliography{bib.bib}

\appendix
\section{Linear vs. General Contracts} \label{app:independent_lin_opt_ratio}
In this appendix we study the gap between the principal's utility in a general contract vs. linear contract. 
We quantify this gap by the worst-case ratio (over all instances) between the principal's utility in the optimal general and linear contracts.

Our main result here is a lower bound of $\Omega(n)$ on this gap, even for $m=3$ outcomes\footnote{$m=3$ is minimal, as $m=2$ is the binary outcome setting, where linear contracs are optimal.}:

\begin{theorem} \label{thm:lin_optimal_ratio}
    There exists an instance with $n$ independent actions and $m=3$ outcomes such that the worst-case ratio between the optimal principal's utility in a general contract and in a linear contract is $\Omega(n)$.
\end{theorem}

\begin{proof}
Consider an instance with an action set $\actions = [n]$, where action $i\in \actions$ has cost $c_i=\frac{2^i-i}{2^{n+1}}$. 
The outcomes are $\outcomes = \{1, 2, 3\}$, with $r(1)=0$ and $r(2), r(3)=1$. The probabilities are given by
\[
p_{ij} = \begin{cases}
    1-2^{i-n-1} & \text{if } j=1 \\
    c_i & \text{if } j=2 \\
    2^{i-n-1} - c_i & \text{if } j=3.
\end{cases} 
\]
We first establish an upper bound on the principal's utility from a linear contract. Fix some $\alpha \in [0,1]$, and let $t_\alpha$ be the corresponding linear contract.

Let $i\in \actions$ be some action, and consider the reservation value $z_i(\alpha)$, when it is non-negative (the case where it is negative is irrelevant since it implies the agent never takes action $i$). According to Equation~\eqref{eq:res_value} we have
\[
(p_{i2} + p_{i3})(\alpha - z_i(\alpha)) =\E[(\alpha \cdot r(X_i) - z_i(\alpha))^+] = c_i
\]
which implies $z_i(\alpha) = \alpha - \frac{c_i}{p_{i2}+p_{i3}}$. We now that any agent's optimal strategy is to take actions in non-decreasing order of $\frac{c_i}{p_{i2} + p_{i3}} = \frac{(2^i - i)/2^{n+1}}{2^{i-n-1}}=1-i\cdot 2^{-i}$, i.e., in order $1,\dots, n$, until one of them reveals a non-zero outcome or $\frac{c_i}{p_{i2}+p_{i3}} > \alpha$. In fact, this defines a unique optimal strategy. 
Let $\hat{i} = \max_{i\in \actions: 1-i\cdot2^{-i} \le \alpha} i$, i.e., the last action the agent is willing to take. The principal's utility is given by
\[
\begin{split}
u_P(t_\alpha) &= (1-\alpha)\left(1-\prod_{i=1}^{\hat{i}}p_{i0}\right) = (1-\alpha)\left(1-\prod_{i=1}^{\hat{i}}\left(1-2^{i-n-1}\right)\right) \\
&\le (1-\alpha)\left(1-\left(1-\sum_{i=1}^{\hat{i}}2^{i-n-1}\right)\right) ~\le~ (1-\alpha)2^{\hat{i}-n} ~\leq~ (\hat{i}\cdot 2^{-\hat{i}})2^{\hat{i}-n}~=~ \hat{i} 2^{-n} ~\le~ \frac{n}{2^n},
\end{split}
\]
where we've used that $1-\alpha \leq \hat{i} \cdot 2^{-\hat{i}}$ by the definition of $\hat{i}$.

We next establish a lower bound on the principal's utility from the optimal general contract.
Fix some $\varepsilon \in (0,1)$, and let $t$ be the contract defined by
\[
t(j) = \begin{cases}
    0 & \text{if }j=1 \\
    1 & \text{if }j=2 \\
    \varepsilon & \text{if }j=3.
\end{cases}
\]
Now, for any action $i\in \actions$, the reservation value $z_i(t)$ is given by 
\[
\begin{split}
c_i &= \E[(t(X_i) - z_i(t_\varepsilon))^+] = p_{i2} (t(2) - z_i(t))^+ + p_{i3} (t(3)-z_i(t))^+ \\
&=  c_i(1-z_i(t))^+ + p_{i3} (\varepsilon - z_i(t))^+.
\end{split}
\]
The only way to satisfy this equation is with $z_i(t) \in (0, \varepsilon)$, which gives us the equality
\[
\begin{split}
c_i &= c_i(1-z_i(t)) + p_{i3}(\varepsilon - z_i(t)) = c_i(1-z_i(t)) + (2^{i-n-1}-c_i)(\varepsilon - z_i(t)) \\
&= c_i + (2^{i-n-1}-c_i)\varepsilon - 2^{i-n-1} z_i(t).
\end{split}
\]
This simplifies to
\[
z_i(t_\varepsilon) = \varepsilon\left(1-\frac{c_i}{2^{i-n-1}}\right) = \varepsilon \cdot I \cdot 2^{-i}.
\]
According to \cite{weitzman1978optimal}, the agent takes actions in a non-increasing order of $z_i(t)$, i.e.,  in order $1, \dots, n$, and halts when one of them reveals a non-zero outcome. This implies that the principal's utility is given by
\[
\begin{split}
u_P(t) &= \sum_{i=1}^n\left(\prod_{k=1}^{i-1}p_{k1}\right)(p_{i2}(1-t(2))+p_{i3}(1-t(3))))=(1-t(3))\sum_{i=1}^n\left(\prod_{k=1}^{i-1} \left(1- 2^{k-n-1}\right)\right) p_{i3} \\
&\ge (1-\varepsilon)\sum_{i=1}^n \left(1-\sum_{k=1}^{i-1}  2^{k-n-1}\right) p_{i3}\ge (1-\varepsilon) \sum_{i=1}^n (1-2^{i-n-1})(2^{i-n-1}-c_i) \\
&= (1-\varepsilon)\sum_{i=1}^n (1-2^{i-n-1})\frac{i}{2^{n+1}}\ge (1-\varepsilon) \frac{1}{2^{n+1}}\sum_{i=1}^n\frac{1}{2}i = \Omega\left(\frac{n^2}{2^n}\right).
\end{split}
\]
To conclude, the principal's utility from any linear contract is upper bounded by $O\left(\frac{n}{2^n}\right)$, while the utility from a general contract is lower bounded by $\Omega\left(\frac{n^2}{2^n}\right)$.
This gives a gap of $\Omega(n)$, as desired.
\end{proof}

\paragraph{An upper bound.}
We note that this gap is upper bounded by $O(n^2m)$.
This upper bound is based on a result by \citet{dutting2019simple}, showing that the ratio between the optimal principal's utility in a general contract and in a linear contract is upper bounded by the number of critical values of $\alpha$ in the corresponding linear contract. 
Combined with our  Proposition~\ref{prop:polymany_critvalues}, this implies an upper bound of $O(n^2 m)$ on this ratio.

\section{Superpolynomial Number of Best Responses 
} \label{app:brset_lower_bound}

In this section, we show that there exist instances with independent actions that admit super-polynomially many best responses of the agent, when $m$ is not a constant.
Specifically, we prove the following proposition:
\begin{proposition} \label{prop:superpoly_brset}
    For any $m \ge 2$, and any $n\in \mathbb{N}$, there exists an instance of the independent-action sequential contract problem with $n$ actions and $m$ outcomes such that there are $n^{\Omega(m)}$ best responses to monotone contracts.
\end{proposition}

\begin{proof}
Let $m \ge 2$ and let $n\in \mathbb{N}$. Assume without loss of generality\footnote{Otherwise we can take the maximal $n' < n$ such that $m-1$ divides $n'$, and add $n-n'$ ``dummy'' actions.} that $m-1$ divides $n$, and let $\ell = \frac{n}{m-1}$. 
We construct the following instance:
\begin{enumerate}
    \item The set of outcomes is $[m]$, and for any outcome $j\in \{2, \dots, m\}$, the reward is $r(j) = \ell^j$. 
    \item The set of actions is $\actions = \{2, \dots, m\} \times [\ell]$, and for every action $a=(j, i) \in A$, $c_{a} = i \cdot r(j)$.
    \item For any action $a=(j, i)$ and any outcome $j'\in \outcomes$, $p_{aj'}$ is given by
    \[
    p_{aj'} = \begin{cases}
        \frac{1}{2} & \mbox{if } j'=1 \lor j' = j\\
        0 & \mbox{else}.
    \end{cases}
    \]
\end{enumerate}
Essentially, this definition allows us the following two degrees of freedom in the best response, by a careful choice of contracts. Firstly, for any outcome $\{2, \dots, m\}$, we can choose how many actions from $\{j\} \times [\ell]$ are taken with non-zero probability. Secondly, we can choose the order in which the actions $\{2, \dots, m\}\times \{1\}$ are taken. When $m$ is small, we will use the former, and when $m$ is large, the latter. We begin our case by case analysis:

\paragraph{Case 1: $m \le \sqrt{n}$.} For any vector $v\in [\ell]^{m-1}$, let $t_v$ denote the following contract:
\[
t_v(j) = \begin{cases}
    0 & \mbox{if } j=1 \\
    (2+\frac{1}{2\ell})r(j) \cdot v_{j-1} & \mbox{if } j>1.
\end{cases}
\]
Note that for any $v\in [\ell]^{m-1}$, this contract is monotone. 
We claim that for every $v\in [\ell]^{m-1}$, the agent's best response to $t_v$, is different. This follows from the following observation, which means that the set of actions taken with non-zero probability is different.
\begin{observation}
    Let $v\in [\ell]^{m-1}$.  Consider the agent's best response to $t_v$. Let $a=(j,i)$ be some action, it is easy to see that if $i \le v_{j-1}$ then $z_{a}(t_v) > 0$, meaning action $a$ is taken with non-zero probability, and if $i > v_{j-1}$ then $z_{a}(t_v) < 0$, meaning action $a$ is never taken.
\end{observation}
This implies that the number of best responses is lower bounded by
\[
\ell^{m-1} \ge (\sqrt{n})^{m-1} = n^{\Omega(m)},
\]
as needed.
\paragraph{Case 2: $m > \sqrt{n}$.} For any bijection $\rho: [m-1]\rightarrow \{2, \dots, m\}$, denote by $t_\rho$ the following contract:
\[
t_\rho(j) = \begin{cases}
    0 & \mbox{if } j=1 \\
    \rho^{-1}(j) + 2 r(j) & \mbox{if } j>1.
\end{cases}
\]
Note that for any bijection $\rho$, $t_\rho$ is monotone.
We claim that for each $\rho$ the agent's best response to $t_\rho$ has a different action-taking order. This is implied by the following observation:
\begin{observation}
    Let $\rho: [m-1]\rightarrow \{2, \dots, m\}$ be a bijection. It holds that
    \[
    0 < z_{(\rho(1), 1)} (t_\rho) < \dots < z_{(\rho (m-1), 1)} (t_\rho).
    \]
\end{observation}

This implies that the number of best responses is lower bounded by
\[
(m-1)! \ge \left(\frac{1}{2}\sqrt{n}\right)^{\frac{1}{2} m} = n^{\Omega(m)},
\]
as needed.
\end{proof}

\section{Polynomial Time Algorithm For the Principal's Utility Under Independent Actions} \label{app:independent_principal_utility}
In this appendix we prove Proposition~\ref{prop:principal_utility}.
\princutility*

The first challenge we address is that while finding an agent's optimal strategy is easy, it is unclear how to efficiently find an agent's optimal strategy that leads to the best principal's utility.
To overcome this, for a given contract $t$, we devise a contract $\epsiloncontract$ which is sufficiently close to $t$ such that maximizing the agent's utility under $\epsiloncontract$ also maximizes his utility under $t$,  but also causes the agent's utility itself to break ties in favor of the principal. This contract is established in the following proposition, and we later analyze its properties.
\begin{proposition} \label{prop:varepsiloncontract}
    For a small enough $\varepsilon > 0$, it holds that $\epsiloncontract(j) = t(j) + \varepsilon \cdot (r(j) - t(j))$ is a valid contract (i.e., all payments are non-negative), and any strategy $\strat$ which maximizes the agent's utility under $\epsiloncontract$ also maximizes his utility under $t$.
\end{proposition}
\begin{proof}
    It is easy to see that $\epsiloncontract$ is a valid contract for a small enough $\varepsilon>0$, Indeed, if $t(j)$ is positive, for a small enough $\varepsilon$, it holds that $\epsiloncontract(j)$ is non-negative. On the other hand, if $t(j) = 0$, then $\epsiloncontract(j) \ge t(j) \ge 0$.

    From Observation~\ref{obs:agent_opt_strat_cond}, it suffices to show that for a small enough $\varepsilon$ the following conditions are satisfied:
    \begin{enumerate}
        \item For any $i_1, i_2$ such that $z_{i_1}(t) < z_{i_2} (t)$ it holds that $z_{i_1}(\epsiloncontract) < z_{i_2} (\epsiloncontract)$.
        \item For any $i, j$ such that $t(j) > z_i(t)$ it holds that $\epsiloncontract(j) > z_i(\epsiloncontract)$.
        \item For any $i, j$ such that $t(j) < z_i(t)$ it holds that $\epsiloncontract(j) < z_i(\epsiloncontract)$.
        \item For any $j_1, j_2$ such that $t(j_1) < t(j_2)$ it holds that $\epsiloncontract(j_1) < \epsiloncontract(j_2)$.
    \end{enumerate}
    From the continuity of $z_i(t)$, it is immediate that these conditions hold for a small enough $\varepsilon$.
\end{proof}

Furthermore, a small enough $\varepsilon$ that satisfies Proposition~\ref{prop:varepsiloncontract} can be found. The following lemma bounds the change in $z_i(t)$ as $t$ increases/decreases slightly, which directly implies how to find an $\varepsilon$ that satisfies the inequalities described in the proof of Proposition~\ref{prop:varepsiloncontract}.
\begin{lemma}
    For any two contracts $t, t'$ and any $\delta > 0$,
    if for any outcome $j\in \outcomes$ it holds that
    \[
    t(j) - \delta \le t'(j) \le t(j)+\delta,
    \]
    then it also holds that for any action $i\in \actions$
    \[
    z_i(t) - \delta \le z_i(t') \le z_i(t) + \delta.
    \]
\end{lemma}
\begin{proof}
    Let $i\in \actions$ be some action.
    By the definition of the reservation value (see Equation~\eqref{eq:res_value}), $z_i(t')$ is such that satisfies
    \[
    \E[(t'(X_i) - z_i(t'))^+] = c_i,
    \]
    and $z_i(t)$ satisfies
    \[
    \E[(t(X_i) - z_i(t))^+] = c_i.
    \]
    We start with the lower bound. To show that $z_i(t') \ge z_i(t) - \delta$ it suffices to show that the inequality $\E[(t'(X_i)-(z_i(t)-\delta))^+] \ge c_i$ holds. Indeed,
    \[
    \E[(t'(X_i)-(z_i(t)-\delta))^+] \ge \E[(t(X_i) - \delta -z_i(t)+\delta)^+] = c_i.
    \]
    For the upper bound, we show that $\E[(t'(X_i)-(z_i(t) + \delta))^+] \le c_i$. Similarly to before we have
    \[
    \E[(t'(X_i)-(z_i(t)+\delta))^+] \le \E[(t(X_i) + \delta -z_i(t)-\delta)^+] = c_i,
    \]
    concluding the proof.
\end{proof}

The following lemma shows an important relation between the agent's utilities in $t$ and $\epsiloncontract$, and the principal's utility in $t$. This relation is the reason we examine $\epsiloncontract$.
\begin{lemma} \label{lemma:epsiloncontract_utility}
    For any contract $t$, any $\varepsilon > 0$ such that $\epsiloncontract$ is a valid contract, and any agent's strategy $\strat$ it holds that
    \[
    u_A(\epsiloncontract, \strat) = \varepsilon \cdot u_P(t, \strat) + u_A(t, \strat).
    \]
\end{lemma}
\begin{proof}
    Recall that $S(\strat)$ denotes the random variable depicting the set of actions taken by the agent under $\strat$, and $\omega(\strat)$ is the random variable depicting the project's final outcome under $\strat$ (see Section~\ref{sec:model}).
    \[
    \begin{split}
    u_A(\epsiloncontract, \strat) &= \E[\epsiloncontract(\finaloutcome(\strat)) - \sum_{i\in S(\strat)} c_i] = \E[\varepsilon\cdot (r(\finaloutcome(\strat)) - t(\finaloutcome(\strat)))] + \E[t(\finaloutcome(\strat)) - \sum_{i\in S(\strat)} c_i] \\
    &= \varepsilon \cdot u_P(t, \strat) + u_A(t, \strat),
    \end{split}
    \] 
    as needed.
\end{proof}

As a corollary of Proposition~\ref{prop:varepsiloncontract} and Lemma~\ref{lemma:epsiloncontract_utility} we get
\begin{corollary}
    For any contract $t$, let $\epsiloncontract$ be per Proposition~\ref{prop:varepsiloncontract}. Any agent's optimal strategy $\strat$ under the contract $\epsiloncontract$ satisfies $u_P(t, \strat) = u_P(t)$.
\end{corollary}

To address the challenge of computing the principal's utility under some strategy, we consider a simple type of agent's strategies for which we can efficiently compute the principal's utility.
\begin{definition} \label{def:non_adaptive_strat} 
    An agent's strategy $\strat$ is \emph{non-adaptive} if it can be described by a tuple $(\sigma, \rho, \tau)$, where 
    \begin{enumerate}
        \item $\sigma:[n] \rightarrow \actions$ is a bijection which defines the action taking order in the following way: if $\strat$ doesn't halt before it, the $i$th action is taken is $\sigma(i)$.
        \item $\rho: \outcomes\rightarrow [m]$ is a bijection describing the agent's preference of outcomes, i.e., the agent prefers outcome $j_1$ over $j_2$ if $\rho(j_1) > \rho(j_2)$.
        \item $\tau = \{\tau_i\}_{i\in \actions}\subseteq [m]$ defines the halting condition in the following way: before taking an action $i \in \actions$, $\strat$ halts if his current preferred outcome $j$ satisfies $\rho(j) \ge \rho(\tau_i)$.
        \item Upon halting $\strat$ chooses its preferred outcome according to $\rho$ as the project's final outcome.
    \end{enumerate}
\end{definition}
By Observation~\ref{obs:agent_opt_strat_cond} it is easy to see that a non-adaptive optimal strategy exists, and can be easily found\footnote{Let $\sigma$ describe some ordering of the actions such that the reservation value is non-increasing, Let $\rho$ be some ordering that is non-increasing, and let $\tau_i\in \argmax_{j\in [m]: t(j)<z_i(t)}t(j)$.}.

For any non-adaptive strategy $\strat= (\sigma, \tau, \mu)$ one can compute the probability of the final outcome being $j$ for any $j\in \outcomes$ by iterating over the actions taken, according to the order $\sigma$, while maintaining the probability of $j$ being the agent's preferred revealed outcome so far. From this distribution over the outcomes, the principal's utility is immediate. This provides an $O(n m^2)$-time algorithm for computing the principal's utility from a contract. 

\section{Additional Proofs from Section~\ref{sec:independent}} \label{app:independent_proofs}
In this appendix we provide technical details on proofs from Section~\ref{sec:independent}.

\subsection{Proof of Lemma~\ref{lem:res_piecwise_lin}} \label{app:proof_lem_res_piecewise_lin}
In this section we present the proof of Lemma~\ref{lem:res_piecwise_lin},
\respiecwiselin*

\begin{proof}
    If $c_i = 0$, $z_i(\alpha)=\infty$, and we are done.
    
    Otherwise, if $c_i > 0$,
    for any $\alpha\in [0,1]$, the random variable $t_\alpha(X_i)$ attains values within $\{\alpha \cdot r(j)\}_{j\in \outcomes}$, and equals $\alpha \cdot r(j)$ with probability $p_{ij}$.
    Thus, $z_i(\alpha)$ is given by
    \begin{equation} \label{eq:res_lin_contract}
    \sum_{j\in \outcomes: \alpha r(j) >z_i(\alpha)} p_{ij}(\alpha \cdot r(j) - z_i(\alpha)) = c_i.
    \end{equation}
    
    Now, for any $j\in \outcomes$ we define $\alpha_{i,j} = \inf\{\alpha \in [0,1] \mid z_i(\alpha) \ge \alpha \cdot r(j)\}$, which we interpret as $\infty$ if the set is empty. Intuitively, $\alpha_{i,j}$ is the threshold of $\alpha$ from which the agent is willing to execute action $i$ when he has already revealed outcome $j$. We also denote $\alpha_{i,0} = 0$.  Observe that $\alpha_{i,0} \le \alpha_{i,1} \le \dots \le \alpha_{i,m}$, and that $\alpha_{i,m} = \infty$ (since $z_i(\alpha)$ can never be greater than $\alpha \cdot r(m)$ for $i$ such that $c_i>0$).
    
    For any $j\in [m]$, such that $\alpha_{i,(j-1)} \ne \alpha_{i,j}$ we consider the interval $[\alpha_{i,(j-1)}, \alpha_{i,j})$. 
    For any $\alpha \in [\alpha_{i,(j-1)}, \alpha_{i,j})$, the set $\{j\in \outcomes \mid \alpha \cdot r(j) > z_i(\alpha)\}$ is precisely $\{j, \dots, m\}$. Plugging this into Equation~\eqref{eq:res_lin_contract} gives us 
    \[
    \sum_{k=j}^m p_{ik}(\alpha \cdot r(k) - z_i(\alpha)) = c_i,
    \]
    implying that within the interval $[\alpha_{i,(j-1)}, \alpha_{i,j})$, 
    \[
    z_i(\alpha) = \frac{\sum_{k=j}^m p_{ik} r(k)}{\sum_{k=j}^m p_{ik}} \cdot \alpha - \frac{c_i}{\sum_{k=j}^m p_{ik}}.
    \]
    Therefore, $z_i(\alpha)$ is piecewise linear and its slope is monotonically increasing.
    
    We have shown at most $m$ intervals, within each of which $z_i(\alpha)$ is linear, and such that from one interval to the next the slope is weakly increasing, which proves our desired piecewise linearity and convexity.
    Since we've given an explicit expression for $z_i(\alpha)$ within each interval $[\alpha_{i,(j-1)}, \alpha_{i,j})$, the only missing piece for an explicit expression for $z_i(\alpha)$ is explicit closed-form expressions for $\alpha_{i,j}$ for any $j \in \outcomes$. Indeed, an explicit expression for $\alpha_{i,j}$ follows from the continuity of $z_i(\alpha)$. By definition of $\alpha_{i,j}$, because $z_i(\alpha)$ is continuous, the following equality holds
    \[
    z_i(\alpha) = \alpha_{i,j} \cdot r(j),
    \]
    which, in combination with Equality~\ref{eq:res_lin_contract}, gives
    \[
    \alpha_{i,j} = \frac{c_i}{\sum_{k=j}^m p_{ik}(r(k) - r(j))}.
    \]
    Note that all of our closed-form expressions are efficiently computable, concluding the proof.
\end{proof}

\subsection{Proof of Lemma~\ref{lemma:crit_value_cond}} \label{app:proof_lemma_crit_value_cond}
In this section we provide the proof for Lemma~\ref{lemma:crit_value_cond}.
\critvaluecond*

\begin{proof}
    Let $\alpha'\in (0,1)$ be some critical value. Denote $\strat = \bestresponse(t_{\alpha'})$. According to Observation~\ref{obs:crit_value_opt_strat}, $\strat$ isn't an optimal strategy directly to the left of $\alpha'$.
    Consider which condition of Observation~\ref{obs:agent_opt_strat_cond} is violated by $\strat$ directly to the left of $\alpha'$.

    Assume condition (1) of Observation~\ref{obs:agent_opt_strat_cond} is violated by $\strat$ to the left of $\alpha'$, i.e., there exist actions $i_1\ne i_2\in \actions$ such that with non-zero probability when action $i_1$ is taken, $i_2$ is still untaken and $z_{i_1}(\alpha) < z_{i_2}(\alpha)$. Because $\strat$ is optimal at $\alpha'$, we have $z_{i_1}(\alpha') \ge z_{i_2}(\alpha')$, which implies (from continuity) the first condition of Lemma~\ref{lemma:crit_value_cond} holds, as needed.

    Assume condition (2) of Observation~\ref{obs:agent_opt_strat_cond} is violated,  i.e., with non-zero probability when $j\in \outcomes$ is a revealed outcome of maximum payment and $i\in \actions$ is an untaken action of maximum $z_i(\alpha)$, it holds that either $\alpha \cdot r(j) < z_i(\alpha)$ and $\strat$ halts or $\alpha \cdot r(j) > z_i(\alpha)$ and $\strat$ doesn't halt.
    If $\alpha \cdot r(j) < z_i(\alpha)$ and $\strat$ halts, then because $\strat$ is optimal at $\alpha'$, we have $\alpha' \cdot r(j) \ge z_i(\alpha')$, which implies the second condition of Lemma~\ref{lemma:crit_value_cond}, as needed.
    Otherwise, if $\alpha \cdot r(j) > z_i(\alpha)$ and $\strat$ doesn't halt, then because $\strat$ is optimal at $\alpha'$ we have $\alpha'\cdot r(j) \le z_i(\alpha')$, which also implies the second condition of Lemma~\ref{lemma:crit_value_cond}, as needed.
    
    Finally, assume towards contradiction that condition (3) of Observation~\ref{obs:agent_opt_strat_cond} is violated to the left of $\alpha'$, i.e. with non-zero probability upon halting $\alpha \cdot r(\omega(\strat)) < \alpha \cdot r(j)$ for some revealed outcome $j$. This implies $r(\omega(\strat)) < r(j)$, which implies $\alpha' \cdot r(\omega(\strat)) < \alpha' \cdot r(j)$, contradicting $\strat$ being optimal at $\alpha'$.
\end{proof}

\subsection{Proof of Lemma~\ref{lem:polytope_max}} \label{app:proof_lemma_polytope_max}
In this section we prove Lemma~\ref{lem:polytope_max}.
\polytopemax*
\begin{proof}
    We begin by noting that the principal's utility is bounded(for example by the agent's utility under the linear contract $t_1$), therefore it has a supremum in $Q$. We denote this supremum by $u$, i.e., $u=\sup\{u_P(t)\mid t\in Q\}$.
    
    Let $\{t^\ell\}_{\ell \in \mathbb{N}}$ be a sequence in $Q$ with utilities that converge to $u$. Since $Q$ is a polytope, and therefore bounded, the sequence $\{t^\ell\}_{\ell \in \mathbb{N}}$ has a converging subsequence, assume without loss of generality that it is $\{t^\ell\}_{\ell \in \mathbb{N}}$ itself, and denote its limit by $t'$. By showing $u_P(t') \ge u$, we are done. 
    
    Because the number of possible agent's strategies is finite, there exists a strategy $\strat$ and subsequence of $\{t^\ell\}_{\ell \in \mathbb{N}}$ in which the agent's best response is always $\strat$. Assume without loss of generality that this subsequence is $\{t_\ell\}_{\ell \in \mathbb{N}}$ itself.
    Since for any strategy $\strat'$ it holds that $u_A(t, \strat')$ is continuous (in $t$), and since $\strat$ maximizes the agent's utility in every $t_\ell$, it also maximizes the agent's utility in $t'$, which due to tie-breaking in favor of the principal means
    \[
    u_P(t') = u_P(t', \bestresponse(t')) \ge u_P(t', \strat) = u,
    \]
    where the last equality is due to the continuity (in $t$) of $u_P(t, \strat$).
\end{proof}

\subsection{Proof of Lemma~\ref{lem:polytope_boundary}} \label{app:proof_lemma_polytope_boundary}
In this section we prove Lemma~\ref{lem:polytope_boundary}
\polytopeboundary*

\begin{proof}
    Let $t^0\in Q$ be such that maximizes the principal's utility in $Q$. If $t^0\notin \relint(Q)$, we are done. Otherwise, we show a contract $t^1$ in the relative boundary of $Q$ with utility $u_P(t^1) \ge u_P(t^0)$. 
    Note that for any agent's strategy $\strat$, the principal's utility from a contract $t$ when the agent engages in strategy $\strat$ is given by the following linear (in $t$) expression:
    \[
    u_P(t, \strat)  = \sum_{j\in \outcomes} P_j (r(j)-t(j)),
    \]
    where $P_j$ is the probability of the project's outcome being $j$ when the agent engages in strategy $\strat$.
    It is known that a linear expression is maximized at a vertex of a polytope, meaning $u_P(t,\bestresponse(t^0))$ is maximized at some $t^1$ which is a vertex of $Q$ and therefore in its relative boundary. 
    Since $\bestresponse(t^0)$ maximizes the agent's utility in $t^0 \in \relint(Q)$, due to our assumption that $(\actionspreorder, \outcomespreorder, \haltgreater, \haltsmaller)$ is constant in $\relint(Q)$, we get that $\bestresponse(t^0)$ maximizes the agent's utility everywhere in $\relint(Q)$.
    Because, when fixing any strategy $\strat$, the agent's utility $u_A(t,\strat)$ is continuous in $t$, we get that $\bestresponse(t_0)$ also maximizes the agent's utility in $t_1$. Due to tie-breaking in favor of the principal, it holds that
    \[
    u_P(t^1) = u_P(t^1, \bestresponse(t^1)) \ge u_P(t^1, \bestresponse(t^0)) \ge u_P(t^0, \bestresponse(t^0)) = u_P(t^0),
    \]
    as needed.
\end{proof}

\subsection{Proof of Proposition~\ref{prop:final_hyperplane_arrangment}} \label{app:hyperplane_arrangement_correctness}
In this section we prove Proposition~\ref{prop:final_hyperplane_arrangment}
\finalhyperplanearrangment*
\begin{proof}
    Let $F$ be some face of the hyperplane arrangement such that $F\subseteq \contractregion$, and let $t^1, t^2 \in \relint(F)$. 
    
    Assume by contradiction that $\outcomespreorder(t^1) \ne \outcomespreorder(t^2)$. This means there exist some $j_1, j_2 \in \outcomes$ such that $j_1 \le_{\outcomespreorder(t^1)} j_2$ but $j_1 \not\le_{\outcomespreorder(t^2)} j_2$. Equivalently, $t^1(j_1) \le t^1(j_2)$ but $t^2(j_1) > t^2(j_2)$. This implies the existence of $t\in \relint(F)$ such that $t(j_2) = t(j_1)$, meaning $t$ is in $E_{j_1, j_2}$ but $t^2$ is not, contradicting Observation~\ref{obs:face_relint_cond}.

    We have that $\outcomespreorder$ is constant within $\relint(F)$, which means there exists some bijection $\rho: [m] \rightarrow[m]$ such that 
    \[
    t(\rho(1)) \le \dots \le t(\rho(m))
    \]
    for any $t\in \relint(F)$. 
    
    Assume by contradiction that $\haltgreater(t^1) \ne \haltgreater(t^2)$. This implies the existence of $i\in \actions$ and $j\in \outcomes$ such that $t^1(\rho(j)) \le z_i(t^1)$ but $t^2(\rho(j)) > z_i(t^2)$. Note that $t^1(\rho(j)) \le z_i(t^1)$ implies 
    \[
    \sum_{k=j}^m p_{i\rho(k)} (t^1(\rho(k)) - t^1(\rho(j)) \ge c_i,
    \]
    and $t^2(\rho(j)) > z_i(t^2)$ implies
    \[
    \sum_{k=j}^m p_{i\rho(k)} (t^2(\rho(k)) - t^2(\rho(j)) < c_i.
    \]
    This implies the existence of $t\in \relint(F)$ such that $t\in T_{\rho, i, j}$ but $t^2 \notin T_{\rho, i, j}$, contradicting Observation~\ref{obs:face_relint_cond}.
    
    The proof for $\haltsmaller$ being constant is almost entirely the same as what we did for $\haltgreater$ above.

    Finally, assume by contradiction that there exist $t^1, t^2\in \relint(F)$ with $\actionspreorder(t^1) \ne \actionspreorder(t^2)$. 
    This implies the existence of $i_1, i_2 \in \actions$ such that $z_{i_1} (t^1) \le z_{i_2} (t^1)$ but $z_{i_1} (t^2) > z_{i_2} (t^2)$. From continuity, there exists $t\in \relint(F)$ such that $z_{i_1}(t) = z_{i_2}(t)$, which means $t\in Z_{i_1, i_2, \haltgreater_{i_1}, \haltgreater_{i_2}}$ but $t^2\notin Z_{i_1, i_2, \haltgreater_{i_1}}$, contradiction.
\end{proof}

\subsection{Proof of Claim~\ref{claim:equal_spread_util}} \label{app:equal_spread_util}
In this section we prove Claim~\ref{claim:equal_spread_util}
\equalspreadutil*
\begin{proof}
    By Observation~\ref{obs:agent_best_response}, the agent's best response to the equal-spread contract $t^S$ is to take actions $1$ and $2$, then take action $3$ if and only if none of the outcomes in $S\cup \{k+1\}$ were realized (either $k+1$ in action $1$ or an outcome $j\in S$ in action $2$), and then halt, choosing outcome $k+1$ if it was realized, one of the outcomes in $S$ if one was realized, or the $0$ outcome if none of the outcomes in $S \cup \{k+1\}$ were realized. 
    
    We first turn our attention to the principal's expected reward. It is equal exactly to the probability with which outcome $k+1$ is realized when the agent best responds. This can happen in one of two ways; either outcome $k+1$ is realized from action $1$, which happens with probability $1-\varepsilon$, or action $1$ didn't produce outcome $k+1$, action $2$ didn't produce an outcome in $S$, and action $3$ yielded outcome $k+1$, which happens with probability $\varepsilon (1- \sum_{j\in S} a_j)q$.

    Overall, the principal's expected reward is 
    \[
    1-\varepsilon + \varepsilon(1- x)q.
    \]
    We now turn to the expected payment. The principal pays $\frac{c}{x+q}$ if any of the outcomes in $S\cup \{k+1\}$ were realized, i.e., with probability $1-\varepsilon(1- x)(1-x-q)$.
    Overall, the principal's utility is 
    \[
    u_P(t^S) = 1-\varepsilon + \varepsilon (1- x)q - \left(1-\varepsilon(1- x)(1-x-q)\right)\frac{c}{x+q},
    \]
    as needed.
\end{proof}
\subsection{Proof of Claim~\ref{claim:opt_is_equal_spread}} \label{app:opt_is_equal_spread}
In this section we prove Claim~\ref{claim:opt_is_equal_spread}
\optisequalspread*
\begin{proof}
    Let $t^\star$ denote the optimal contract. We start by claiming that $t^\star(k+1) \ge t^\star(i)$ for any $i\in [k]$. 
    Suppose towards contradiction that there exists some $i\in [k]$ such that $t^\star(i) > t^\star(k+1)$. Then, the social welfare, and thus the principal's utility is upper bounded by $(1-a_i) \cdot r_{k+1} < 1-\varepsilon$ (the last inequality is by our choice of $\varepsilon$). Note that the principal's utility from $t^{[k]}$, as given by Claim~\ref{claim:equal_spread_util} when plugging in $x=0.2$, is greater than $1-\varepsilon$, this is a contradiction. 

    Additionally, note that $z_{3}(t^\star) \ge 0$, since otherwise 
    the social welfare and thus the principal's utility is upper-bounded by $1-\varepsilon$ (by Observation~\ref{obs:agent_best_response}), which is lower than $u_P(t^{[k]})$ (by Claim~\ref{claim:equal_spread_util}).
    Now, define \[
    S = \{j\in [k] \mid t^\star(j) > z_{3}(t^\star)\},
    \]
    we claim that $t^\star = t^S$, thus proving that $t^\star$ is an equal-spread contract. Observe that
    \[
    \begin{split}
    \left(\sum_{j \in S} a_j + q\right) t^\star(k+1) &= \sum_{j \in S \cup \{k+1\}} p_{(k+2), j} t^\star(k+1) \\
    &= \sum_{j \in S} a_j (t^\star(k+1) - t^\star(j)) + \sum_{j \in S \cup \{k+1\}} p_{(k+2), j} t^\star(j)\\
    &= \sum_{j \in S} a_j (t^\star(k+1) - t^\star(j)) + \sum_{j \in S \cup \{k+1\}} p_{(k+2), j} (t^\star(j) - z_{k+2}(t^\star) ) \\ &+ \sum_{j \in S \cup \{k+1\}} p_{(k+2), j}z_{k+2}(t^\star) \\
    &= \sum_{j \in S } a_j (t^\star(k+1) - t^\star(j)) + c + \left(\sum_{j \in S} a_j + q\right)  z_{k+2}(t^\star),
    \end{split}
    \]
    where the last inequality is by the definition of $z_{k+2}(t^\star)$. Rewriting, we get
    \[
    \begin{split}
    t^\star(k+1) &\ge \frac{\sum_{j \in S} a_j (t^\star(k+1) - t^\star(j))}{q+\sum_{j\in S} a_j}+\frac{c}{q+\sum_{j\in S} a_j} + z_{k+2}(t^\star)\\
    &= \frac{\sum_{j \in S} a_j (t^\star(k+1) - t^\star(j))}{q+\sum_{j\in S} a_j}+t^S(k+1)+z_{k+2}(t^\star).
    \end{split}
    \]
    We additionally note that, due to Observation~\ref{obs:agent_best_response}, $t^\star$ and $t^S$ cause the agent to have a best response with the same reward (since they both take the same actions under the same conditions, and choose the good outcome when available).
    This implies that the difference between $u_P(t^S)$ and $u_P(t^\star)$ is simply the difference in expected payments, and since $t^\star$ is optimal we get
    \[
    \begin{split}
    0 &\ge u_P(t^S) - u_P(t^\star) = -\sum_{j\in [k+1]} \Pr[\omega = j] (t^S(j)  t^\star(j)) \\
    &\ge \Pr[\omega = k+1](t^\star(k+1)-t^S(k+1)) - \sum_{j\in S} Pr[\omega = j] (t^S(j)-t^\star(j))\\
    &\ge (1-\varepsilon) (t^\star(k+1)-t^S(k+1)) -  \sum_{j\in S} Pr[\omega = j] (t^S(k+1)-t^\star(j)) \\
    &\ge (1-\varepsilon) (t^\star(k+1)-t^S(k+1)) - \sum_{j\in S} Pr[\omega = j] (t^\star(k+1)-t^\star(j)) \\
    &\ge (1-\varepsilon) (t^\star(k+1)-t^S(k+1)) - \sum_{j\in S} \varepsilon \cdot 2 a_j (t^\star(k+1)-t^\star(j))\\
    &\ge (1-\varepsilon) \frac{\sum_{j \in S} a_j (t^\star(k+1) - t^\star(j))}{q+\sum_{j\in S} a_j} + (1-\varepsilon) z_{k+2}(t^\star) - 2\varepsilon\sum_{j\in S} a_j (t^\star(k+1) - t^\star(j)) \\
    &=
    \left(\frac{1-\varepsilon}{q+\sum_{j\in S} a_j} - 2\varepsilon\right) \sum_{j\in S} a_j (t^\star(k+1) - t^\star(j)) + (1-\varepsilon) z_{3}(t^\star).
    \end{split}
    \]
    Since $\sum_{j\in S} a_j (t^\star(k+1) - t^\star(j)), z_{k+2}(t^\star) \ge 0$, this implies they both equal $0$,  which in turn proves $t^\star = t^S$. This is because $\sum_{j\in S} a_j (t^\star(k+1) - t^\star(j)) = 0$, implies $t^\star(k+1)=t^\star(j)$ for all $j\in S$, and by definition of $S$ we have 
    \[
    0=z_3(t^\star) = \frac{\sum_{j\in S \cup \{k+1\}}p_{3,j} t^\star(j)-c}{\sum_{j\in S} a_j+q} = \frac{t^\star(k+1)\sum_{j\in S \cup \{k+1\}} p_{3,j}-c}{\sum_{j\in S} a_j+q},
    \]
   which gives $t^\star(j) = \frac{c}{\sum_{j\in S} a_j + q} = t^S(j)$ for all $S\cup \{k+1\}$. If we assume by contradiction that $t^\star(j) > 0$ for some $j\notin S\cup \{k+1\}$, we can get a strictly better utility by decreasing it to $0$, thus proving that $t^\star = t^S$, as needed.
\end{proof}

\section{Lower Bound of $\Omega(m)$ on the Number of Critical Points} \label{app:crit_points_m_dependence}
In this section, we show that there exist instances with $2$ independent actions and $\Omega(m)$ critical points. More specifically, we prove the following
\begin{proposition} 
    For any $m \ge 2$, there exists an instance of the independent-action sequential contract problem with $2$ actions and $m$ outcomes such that there are $\Omega(m)$ critical points.
\end{proposition}
\begin{proof}
    Consider the following instance:
    \begin{enumerate}
        \item Action $1$ is free ($c_1 = 0$) and induces a uniform distribution over the outcomes ($p_{1j}=\frac{1}{m}$ for all $j\in [m]$).
        \item Action $2$ has a cost of $\frac{1}{2m}$ ($c_2=\frac{1}{2m}$) and induces a uniform distribution over the outcomes ($p_{2j}=\frac{1}{m}$ for all $j\in [m]$).
        \item For any $j\in [m]$, the reward for outcome $j$ is $j-1$ ($r(j) = j-1$).
    \end{enumerate}
    The agent's best response to a linear contract $\alpha$ is always to take action $1$ first, then take action $2$ if and only if $\alpha (X_{1}-1) \ge z_2(\alpha)$, and finally choose a maximal revealed outcome. Thus, the agent's best response changes whenever $z_2(\alpha) = \alpha (j-1)$ for some $j\in [m]$. Note that this is a ``real'' change in the sense that it changes the agent's behavior with non-zero probability, since with probability $\frac{1}{m}$ after taking action $1$ $X_1=j$.
    Furthermore, note that $z_2(0) < 0$, we claim that $z_2(1) > 1\cdot (m-2)$. Indeed, $z_2(1)$ satisfies Equation~\ref{eq:res_lin_contract}, i.e.,
    \[
    \sum_{j\in [m]: j-1>z_2(1)} p_{2j} (j-1-z_2(1)) = \frac{1}{2m}.
    \]
    If $z_2(1)\le m-2$ we'd have
    \[
    \sum_{j\in [m]: j-1>z_2(1)} p_{2j} (j-1-z_2(1)) \ge p_{2m}(m-1-z_2(1))\ge \frac{1}{m} >\frac{1}{2m},
    \]
    which is a contradiction.
    Overall, we have that $z_2(0)<0$, $z_2(1) > 1(m-1-1)$, and $z_2(\alpha)$ is continuous (from Lemma~\ref{lem:res_piecwise_lin}),
    thus proving that $z_2(\alpha)$ intersects with $\alpha (j-1)$ for all $j=1,\dots,m-1$, yielding at least $m-1$ critical points.
\end{proof}

\section{An Equivalence Between Coverage Functions and ``Correlated Maximum'' Functions} \label{app:generalized_coverage_equiv}
In this section we prove a generalization to Proposition~\ref{prop:coverage_equiv} which may be of independent interest. Specifically we show an equivalence between a set function which maps subsets of correlated random variables to their expected maximum value and coverage functions. In this equivalence the representation size of each is polynomial in the representation size of the other, which implies hardness results with respect to one representation imply hardness results with respect to the other. This equivalence is formalized in the following proposition.

\begin{proposition}
Let $\{X_i\}_{i\in A}$ be (correlated) random variables which attain values within some set $M$ and have support $\mathcal{R}\subseteq M^A$ and probability density function $p: \mathcal{R} \rightarrow [0,1]$. There exists a coverage function  $f=(U, \{w_u\}_{u\in U}, A, h)$ with $|U| \le |\mathcal{R}|^2$ such that $f(S) = \E[\max_{i\in S} X_i]$.
Moreover, for every coverage function $f=(U, \{w_u\}_{u\in U}, A, h)$, there exist (correlated) random variables $\{X_i\}_{i\in A}$, with support $\mathcal{R}\subseteq \{0, L\}^A$, for some $L \ge 0$, such that $f(S) = \E[\max_{i\in S} X_i]$ and $|\mathcal{R}| \le |U|$.
\end{proposition}
\begin{proof}
     Let $\{X_i\}_{i\in A}$ be (correlated) random variables which attain values within some set $M$ and have support $\mathcal{R}\subseteq M^A$ and probability density function $p: \mathcal{R} \rightarrow [0,1]$. For any $m\in M$ we denote 
    \[
    \prev(m) = \max \{x \mid x\in M : x<m\},
    \]
     which we interpret as $0$ when the set is empty.
     
     We now define the coverage function $f=(U, \{w_u\}_{u\in U}, A, h)$, where $U=\mathcal{R}\times M$, for each $u=(v, m)\in U$ we define $w_u = p(v) (m-\prev(m))$, and $h(i) = \{(v, m)\in U \mid v_i \ge m\}$.
    Note that $|M| \le |\mathcal{R}|$, which means $|U| = |\mathcal{R}|\cdot |M| \le |\mathcal{R}|^2$ as needed. 
     Now,
    \[
    \begin{split}
    \E[\max_{i\in S} X_i] &= \sum_{v\in \mathcal{R}} p(v) \max_{i\in S} v_i = \sum_{v\in \mathcal{R}} p(v) \sum_{m \le \max_{i\in S} v_i} (m - \prev(m))\\
    &= \sum_{v\in \mathcal{R}}  \sum_{m \le \max_{i\in S} v_i} p(v) \ (m - \prev(m)) = \sum_{v\in \mathcal{R}}  \sum_{m \le \max_{i\in S} v_i} w_{(v,m)}\\
    &= \sum_{(v,m)\in U} w_{v,m} \indicator[\exists i\in S. ~v_i \ge m] = \sum_{(v,m)\in U} w_{v,m} \indicator[\exists i\in S.~ (v,m)\in h(i)] = f(S),
    \end{split}
    \]
    as needed.

    In the other direction, let $f=(U, \{w_u\}_{u\in U}, A, h)$. Denote $L = \sum_{u\in U} w_u$.
    We define our variables $\{X_i\}_{i\in A}$ by sampling an element $x\in U$, where each element $u\in U$ is sampled with probability
    \[
    Pr[x=u] = w_u\frac{1}{L}.
    \]
    For any $i\in A$ and $u\in U$ we denote $X_i(u) = L\cdot \indicator[u\in h(i)]$. Our random variables $\{X_i\}_{i\in A}$ are defined as $X_i = X_i(x)$, for any $i\in A$. Now:
    \[
    f(S) = \sum_{u\in U} w_u \indicator[\exists i\in S.~u\in h(i)] = \sum_{u\in U} Pr_x[x = u]\cdot L \cdot \indicator[\exists i\in S. X_i(u) = L] = \E_x[\max_{i\in S} X_i].
    \]
    Since each $X_i$ is a function of $x$, and there are $|U|$ possible values of $x$ we get that $\mathcal{R}$, which is defined as support of $\{X_i\}_{i\in A}$, has size $|\mathcal{R}|\le |U|$, as needed.
\end{proof}

\section{Omitted Details from the Proof of Theorem~\ref{thm:corr_contract_hardness}} \label{app:correlated}
In this appendix we provide missing details from the proof of Theorem~\ref{thm:corr_contract_hardness}, namely the correctness of the reduction, as cast in Lemma~\ref{lemma: reduction correctness}.
\reductioncorrectness*
Recall the reduction stated in Section~\ref{sec:corr}:
For any $(k,f': 2^{A'} \rightarrow [0,1])$ per Proposition~\ref{prop:coverage_hardness} with the parameters $M=3, \varepsilon=0.001$, where $f' = (U, \{w_u\}_{u\in U}, A', h')$.
Construct the following instance of the sequential contracts problem: $(\actions, \{c_i\}_{i\in \actions}, f)$ defined as follows. Our set of actions is $\actions = A' \cup \{0\}$. Our costs are $c_i = \frac{1.5}{k+1}$ for any $i\in A'$ and $c_0 = 1-\frac{\gamma}{8}$. Our ``correlated OR'' function is $f=(U, \{w_u\}_{u\in U}, A, h)$, where $h$ is defined as 
\[
h(i) = \begin{cases}
    h'(i) & i\ne 0\\
    U & i=0.
\end{cases}
\]

We first note the following Lemma:
\begin{lemma} \label{lem:no_zero}
    For any $\alpha < c_0$, the agent's best response to $t_\alpha$ never contains $0$.
\end{lemma}
\begin{proof}
    Let $\strat = (s_1, \dots, s_\ell)$ be the agent's best response to the contract $t_\alpha$, and assume by contradiction that it contains action $0$. It is clear that  $0$ is the final action, due to our assumption that in the strategy representation the final action must be taken with non-zero probability.
    
    Denote $\strat' = (s_1, \dots, s_{\ell - 1})$, we show that $\strat'$ has a greater utility than $\strat$, leading to a contradiction. Indeed, 
    \[
    \begin{split}
    c(\strat) - c(\strat') &= \sum_{i=1}^\ell (1-f(\{s_1, \dots, s_{i-1}\})) \cdot c_{s_i} - \sum_{i=1}^{\ell-1} (1-f(\{s_1, \dots, s_{i-1})) \cdot c_{s_i} \\
    &=(1-f(\{s_1, \dots, s_{\ell -1}\})) \cdot c_0,
    \end{split}
    \]
    which means
    \[
    \begin{split}
    u_A(t_\alpha, \strat) - u_A(t_\alpha, \strat') &= \alpha \cdot f(\{s_1, \dots, s_{\ell-1} \}) - c(\strat)- (\alpha \cdot f(\{s_1, \dots, s_\ell \}) - c(\strat')) \\
    &= (c_0 - \alpha)\cdot (1- f(\{s_1, \dots, s_{\ell - 1} \})) < 0,
    \end{split}
    \]
    where the last inequality is because the probability of the last action (i.e., $s_\ell$) being taken is non-zero. This contradicts $\strat$ being a best response.
\end{proof}

We now prove the correctness of our reduction.
\begin{proof} [Proof of Lemma~\ref{lemma: reduction correctness}]
We consider the two possible condition of Proposition~\ref{prop:coverage_hardness}:
\paragraph{Case 1: $f'$ satisfies condition (1) from Proposition~\ref{prop:coverage_hardness}. }
We bound the principal's utility from the optimal contract from below by her utility from the contract $t_{\frac{3}{4}}$. 

Let $S$ be a set that satisfies condition $1$ (i.e., $f'(S)=1$ and $|S|=k$). Note that for any set $T\subseteq S$ it holds that $f'(T) = \frac{|T|}{k}$. The upper bound is from subadditivity, and since for any $i\in A'$ it holds that $f'(\{i\})=\frac{1}{k}$, and the lower bound is since, from subadditivity, $f'(T) \ge f'(S) - f'(S\setminus T) \ge 1-\frac{|S|-|T|}{k}$.
Now, let $s_1, \dots, s_k\in \actions$ be some enumeration of the elements of $S$, and consider the strategy $\strat = (s_1, \dots, s_k)$, we claim that this maximizes the agent's utility. 
The expected cost $c(\strat)$ is
\[
\begin{split}
c(\strat) &= \sum_{i=1}^k (1-f(\{s_1, \dots, s_{i-1}\})) c_{s_i} = \sum_{i=1}^k \left(1-\frac{i-1}{k}\right)\frac{1.5}{k+1} =\frac{1.5}{k+1}\left(k-\frac{1}{k}\sum_{i=1}^k(i-1)\right) \\
&= \frac{1.5}{k+1}\left(k-\frac{1}{k}\frac{k(k-1)}{2}\right) = \frac{1.5}{k+1}\frac{2k - (k-1)}{2} = \frac{3}{4}.
\end{split}
\]
This means that the agent's utility from $\strat$ is
\[
u_A(t_{\frac{3}{4}}, \strat) = \frac{3}{4}f(S) - c(\strat) = 0.
\]
Now, let $\strat' = (a_1, \dots, a_\ell)$ be an agent's strategy which doesn't contain 0 (this is without loss of generality due to Lemma~\ref{lem:no_zero} and because $\frac{3}{4} < 1-\frac{\gamma}{8}$). We show that the agent's utility from $\strat'$ is non-positive, which, due to tie-breaking in favor of the principal, means that $\strat$ (or some other strategy which leads to a success with a probability of 1) is the agent's best response.

If $\ell \le k$, then
\[
\begin{split}
c(\strat') &= \sum_{i=1}^\ell (1-f(\{a_1, \dots, a_{i-1}\})) c_{a_i} \ge \sum_{i=1}^\ell \left(1-\frac{i-1}{k}\right)\frac{1.5}{k+1} =\frac{1.5}{k+1}\left(\ell-\frac{1}{k}\sum_{i=1}^\ell(i-1)\right) \\
&= \frac{1.5}{k+1}\left(\ell-\frac{1}{k}\frac{\ell(\ell-1)}{2}\right) = \frac{1.5}{k+1}\left(\frac{\ell (2k-(\ell - 1))}{2k}\right) \ge \frac{1.5\cdot \ell}{2k},
\end{split}
\]
where the first inequality is due to subadditivity, for any $S\subseteq A'$ it holds that $f(S) = f'(S) \le \frac{|S|}{k}$, and the last inequality is because $\ell \le k$.
This means the agent's utility from $\strat'$ is at most
\[
u_A\left(t_{\frac{3}{4}}, \strat'\right) = \frac{3}{4} f(\{a_1, \dots, a_\ell\}) - c(\strat') \le \frac{3}{4} \frac{\ell}{k} - \frac{3}{4}\frac{\ell}{k} = 0,
\]
as needed.

If $\ell > k$, we can bound $c(\strat')$ from below by $c((a_1, \dots, a_k))$, which as argued above is at least $\frac{3}{4}$, which means the agent's utility from $\strat'$ is at most
\[
u_A\left(t_{\frac{3}{4}}, \strat'\right) = \frac{3}{4} f(\{a_1, \dots, a_k\}) - c(\strat') \le \frac{3}{4} \cdot 1 - \frac{3}{4} = 0,
\]
as needed.

Because $\strat$ maximizes the agent's best response, and because the agent breaks ties in favor of the principal, the principal's utility from $t_\frac{3}{4}$ is at least
\[
u_P\left(t_\frac{3}{4}\right) \ge \left(1-\frac{3}{4}\right) f(S) = \frac{1}{4}, 
\]
as needed.
\paragraph{Case 2: $f'$ satisfies condition (2) from Proposition~\ref{prop:coverage_hardness}. }
It is easy to see that the agent's best response to a contract $t_{1-\frac{\gamma}{8}}$ is $(0)$, which means $u_P(t_{1-\frac{\gamma}{8}}) = \frac{\gamma}{8} > 0$. We conclude the proof by showing that under a contract $t_\alpha$ with $\alpha < 1-\frac{\gamma}{8}$, the agent's utility from any non-empty strategy which doesn't include $0$ is negative. This means that the agent's best response to $t_\alpha$ is to do nothing, implying $u_P(t_\alpha) = 0$, as needed.

Let $\alpha < 1-\frac{\gamma}{8}$, and consider some strategy $(s_1, \dots, s_\ell)$ which doesn't include $0$. 

If $\ell < \frac{k}{10}$, the expected cost of strategy $\strat$ is at least 
\[
c(\strat) \ge  \frac{1.5}{k+1}\left(\frac{\ell (2k-(\ell - 1))}{2k}\right) > \frac{1.5\cdot \ell}{2k}\left(\frac{(2-\frac{1}{10}) k}{k+1}\right),
\]
where the first inequality is the same as we've seen in the analysis of case 1. For $k \ge 3$ this implies $c(\strat) > \frac{\ell}{k}$, which means the agent's utility from $\strat$ is bounded by
\[
u_A(t_\alpha, \strat) = \alpha \cdot f(\{s_1, \dots, s_\ell\}) - c(\strat) < 1\cdot\frac{\ell}{k} - \frac{\ell}{k} < 0,
\]
where the inequality is from our earlier observation that $f(T) \le \frac{|T|}{k}$, and from our bound on $c(\strat)$.  

If $\frac{ k}{10} \le \ell \le 3k$, by applying condition (2) we can bound the expected cost $c(\strat)$ from below by
\[
\begin{split}
c(\strat) &= \sum_{i=1}^\ell (1-f(\{s_1, \dots, s_{i-1}\}) c_{s_i} \ge \sum_{i=1}^{\ell} (e^{-(i-1)/k}-\varepsilon) c_{s_i} = \frac{1.5}{k+1} \left(\sum_{i=1}^{\ell} e^{-(i-1)/k}-\ell\varepsilon\right) \\
&=  \frac{1.5}{k+1} \left(\frac{1-e^{-\ell/k}}{1-e^{-1/k}}-\ell\varepsilon\right) > 1.5 \left(\frac{1-e^{-\ell/k}}{(1-e^{-1/k})(k+1)} -3\varepsilon\right) > 1.5\left(\frac{1-e^{-\ell/k}}{1.3}  - 3\varepsilon\right) \\
&> 1-e^{-\ell / k}+\varepsilon,
\end{split}
\]
where the second to last inequality is because $(1-e^{-x})(x+1)<1.3$ for any $x > 0$, and the last inequality is because our choice of $\varepsilon$ is such that for $\ell \ge \frac{k}{10}$ the inequality is satisfied.
This implies
\[
u_A(t_\alpha, \strat) = \alpha \cdot f(\{s_1, \dots, s_\ell\}) - c(\strat) < \alpha (1-e^{-\ell / k}+\varepsilon) - (1-e^{-\ell / k}+\varepsilon) < 0.
\]
If $\ell > 3k$, we can bound the cost of $c(\strat)$ from below by the cost of $c((s_1, \dots, s_{3k}))$, which, as argued above is at least
\[
1.5\left(\frac{1-e^{-3}}{1.3}  - 3\varepsilon\right) > 1,
\]
implying the agent's utility is negative, concluding the proof.
\end{proof}

\end{document}